\documentclass[letterpaper,twoside,%
  twocolumn,aps,pra,floatfix
  ]{revtex4}
\usepackage{mathrsfs}
\usepackage{amsmath}
\usepackage{amssymb}
\usepackage{graphicx}
\usepackage{bm}
\usepackage[unicode=true,pdfusetitle,
 bookmarks=true,bookmarksnumbered=false,bookmarksopen=false,
 breaklinks=false,pdfborder={0 0 1},backref=section,colorlinks=false]
 {hyperref}

\makeatletter

\pdfpageheight\paperheight
\pdfpagewidth\paperwidth

\providecommand{\tabularnewline}{\\}

\providecommand{\tabularnewline}{\\}

\def\OR{\mathop{\rm OR}}

\usepackage{amsthm}

\newtheorem{statement}{Lemma}
\newtheorem{theorem}{Theorem}
\newtheorem{consequence}[theorem]{Consequence}
\newtheorem{example}{Example}

\def\rank{\mathop{\rm rank}}
\def\mod{\mathop{\rm mod}}
\def\wgt{\mathop{\rm wgt}}
\def\deg{\mathop{\rm deg}}
\def\Hn{\mathcal{H}_2^{\otimes n}}

\def\mat#1{\mathcal{#1}}

\newcommand{\ket}[1]{\left\vert{#1}\right\rangle}

\makeatother

\begin{document}

\author{Alexey A. Kovalev}

\affiliation{Department of Physics \& Astronomy, University of California, Riverside,
California 92521, USA}

\author{Leonid P. Pryadko}

\affiliation{Department of Physics \& Astronomy, University of California, Riverside,
California 92521, USA}

\title{Quantum ``hyperbicycle'' low-density parity check codes with finite
rate}
\begin{abstract}
  We introduce a ``hyperbicycle'' ansatz for quantum codes which gives
  the hypergraph-product (generalized toric) codes by Tillich and
  Z\'emor and generalized bicycle codes by MacKay et al.\ as limiting
  cases.  The construction allows for both the lower and the upper
  bounds on the minimum distance; they scale as a square root of the
  block length.  Many of thus defined codes have finite rate and a
  limited-weight stabilizer generators, an analog of classical
  low-density parity check (LDPC) codes.  Compared to the
  hypergraph-product codes, hyperbicycle codes generally have wider
  range of parameters; in particular, they can have higher rate while
  preserving the (estimated) error threshold.
\end{abstract}
\maketitle

\section{Introduction}

Quantum computing can become a reality only when combined with quantum
error correction as quantum information is intrinsically fragile due
to inevitable coupling to the environment
\cite{chuang-2000,chuang-2001,%
  Chiaverini-2004,Martinis-review-2009}.  Quantum error correcting
codes (QECCs)
\cite{shor-error-correct,Knill-Laflamme-1997,Bennett-1996} offer
protection for quantum information; however, often at a high cost in
the number of auxiliary qubits and with technologically difficult
requirements \cite{Knill-error-bound,Rahn-2002,%
  Dennis-Kitaev-Landahl-Preskill-2002,Steane-2003,%
  Fowler-QEC-2004,Fowler-2005,fowler-thesis-2005,%
  Knill-nature-2005,Knill-2005,Raussendorf-Harrington-2007}.  The
optimization then must play an important role since by simplifying the
syndrome measurements one can lower the overhead and perform error
correction in parallel.  Quantum codes with limited stabilizer
generator weights, a quantum analog of classical low-density parity
check (LDPC) codes \cite{Postol-2001,MacKay:IEEE2004} seem to offer a
viable solution here as the simple structure of stabilizer generators
will require fewer ancillas, fewer rounds of syndrome measurements and 
parallelism. Furthermore, by analogy with classical LDPC codes, there
might exist efficient algorithms for encoding and decoding
\cite{MacKay:IEEE2004}.

However, quantum LDPC codes come at a higher price compared to their
classical analogs due to stringent requirements on code parameters
originating from commutativity of stabilizer generators. In fact,
there are no known families of asymptotically good quantum LDPC codes,
or any bounds suggesting their existence. It then becomes an
intriguing question of the best asymptotic properties achievable with
quantum LDPC codes. In such a setting explicit code designs become
important, in particular, for establishing the lower bounds on the
parameters: the number of encoded qubits $K$ and minimum distance $D$
for a given block length $N$ (which also defines the code rate $K/N$),
e.g., given the upper limit on the weight of stabilizer
generators. 

The best-known quantum LDPC codes (that are also local) are Kitaev's
toric codes and related surface codes with the minimum distance
scaling as $\sqrt{N}$
\cite{kitaev-anyons,Dennis-Kitaev-Landahl-Preskill-2002,%
  Raussendorf-Harrington-2007,Bombin:PRA2007}.  Existence of
single-qubit-encoding LDPC codes with the distance scaling as
$\sqrt{N}\log N$ has been proved in
Ref.~\onlinecite{Freedman:feb2002}.  Tillich and Z\'emor proposed a
finite-rate generalization of toric codes \cite{Tillich2009}.  The
construction relates a quantum code to a direct product of hypergraphs
corresponding to two classical binary codes. Generally, thus obtained
quantum LDPC codes have finite rates and the distances that scale as a
square root of the block length.

In one of the first studies of quantum LDPC codes MacKay et
al. \cite{MacKay:IEEE2004} constructed so called bicycle codes.
Numerically, these codes exhibit good decoding properties; however,
the minimum distance of such codes is unknown.  The quantum
hypergraph-product codes \cite{Tillich2009}, on the other hand, are an
example of LDPC codes with known parameters; however, decoding such
codes may be difficult.  The existence of a finite noise threshold,
with and without syndrome measurement errors, has been recently
established for limited-stabilizer-weight quantum hypergraph-product
codes\cite{Kovalev:2012arXiv}, as well as for any such LDPC code family
with the distance scaling as the square root of block length.
These results, however, might not apply to the constructions of quantum LDPC
codes based on finite geometries\cite{Aly:2008,Farinholt:arXiv2012} and to the constructions based on Cayley Graphs\cite{Zemor:arXiv2012}
due to the unbounded weight of stabilizer generators.

In this work, we construct a large family of codes that in the
limiting cases reduce to (generalized) bicycle and hypergraph-product
codes, hence we call them \emph{hyperbicycle}.  We show that
hyperbicycle codes contain new code families with finite rates and the
distances that scale as a square root of the block length.  In
addition, the hyperbicycle construction can improve the rate of the
hypergraph-product codes while preserving the estimated error
threshold.

\section{Preliminaries }

In this section, we define classical and quantum error correcting
codes. We also review some of the known LDPC code constructions.

\subsection{Classical error correcting codes}

A \emph{classical} $q$-ary block error-correcting code $(n,K,d)_{q}$
is a set of $K$ length-$n$ strings over an alphabet with $q$ symbols.
Different strings represent $K$ distinct messages which can be transmitted.
The (Hamming) distance between two strings is the number of positions
where they differ. Distance $d$ of the code ${\cal C}$ is the minimum
distance between any two different strings from ${\cal C}$.

In the case of \emph{linear} codes, the elements of the alphabet must
form a Galois field $\mathbb{F}_{q}$; all strings form $n$-dimensional
vector space $\mathbb{F}_{q}^{n}$. A linear error-correcting code
$[n,k,d]_{q}$ is a $k$-dimensional subspace of $\mathbb{F}_{q}^{n}$.
The distance of a linear code is just the minimum weight of a non-zero
vector in the code, where weight $\wgt(\mathbf{c})$ of a vector $\mathbf{c}$
is the number of non-zero elements. A basis of the code is formed
by the rows of its \emph{generator matrix} $G$. All vectors that
are orthogonal to the code form the corresponding $(n-k)$-dimensional
dual code, its generator matrix is the parity-check matrix $H$ of
the original code.

For a \emph{binary} code ${\cal C}[n,k,d]$, the field is just
$\mathbb{F}_{2}=\{0,1\}$.  For a \emph{quaternary} code $C$, the field
is $\mathbb{F}_{4}=\{0,1,\omega,\overline{\omega}\}$, with
\begin{equation}
\omega^{2}=\omega+1,\quad\omega^{3}=1,\;\,\mathrm{and}\;\,\overline{\omega}\equiv\omega^{2}.\label{eq:F4def}
\end{equation}

For non-binary codes, there is also a distinct class of
\emph{additive} classical codes, defined as subsets of
$\mathbb{F}_{q}^{n}$ closed under addition (in the binary case these
are just linear codes).

A code ${\cal C}$ is cyclic if inclusion
$(c_{0},c_{1},\ldots,c_{n-1})\in{\cal C}$ implies that
$(c_{n-1},c_{0},c_{1},\ldots,c_{n-2})\in{\cal C}$. Codes that are both
linear and cyclic are particularly simple: by mapping vectors to
polynomials in the natural way, $\mathbf{c}\to c(x)\equiv
c_{0}+c_{1}x+\ldots+c_{n-1}x^{n-1}$, it is possible to show that any
such code consists of polynomials which are multiples of a single
generator polynomial $g(x)$, which must divide $x^{n}-1$ (using the
algebra corresponding to the field $\mathbb{F}_{q}$). The quotient
defines the \emph{check polynomial} $h(x)$,
\begin{equation}
h(x)g(x)=x^{n}-1\label{eq:linear-cyclic-condition}
\end{equation}
which is the canonical generator polynomial of the dual code. The
degree of the generator polynomial is $\deg g(x)=n-k$. The
corresponding generator matrix $G$ can be chosen as the first $k$ rows
of the \emph{circulant matrix} $C_n$ formed by subsequent shifts of
the vector that corresponds to $g(x)$,
explicitly:
\begin{equation}
C_{n}=\left(\begin{array}{ccccc}
g_{0} & g_{1} & g_{2} & \ldots & g_{n-1}\\
g_{n-1} & g_{0} & g_{1}\\
g_{n-2} & g_{n-1} & g_{0} &  & \vdots\\
\vdots &  &  & \ddots\\
g_{1} & g_{2} & g_{3} & \ldots & g_{0}
\end{array}\right).\label{eq:circulant}
\end{equation}

\subsection{Quantum stabilizer codes}

Binary quantum error correcting codes (QECCs) are defined on the complex
Hilbert space $\Hn$ where $\mathcal{H}_{2}$ is the complex Hilbert
space of a single qubit $\alpha\left|0\right\rangle +\beta\left|1\right\rangle $
with $\alpha,\beta\in\mathbb{C}$ and $\left|\alpha\right|^{2}+\left|\beta\right|^{2}=1$.
Any operator acting on such an $n$-qubit state can be represented
as a combination of Pauli operators which form the Pauli group $\mathscr{P}_{n}$
of size $2^{2n+2}$ with the phase multiplier $i^{m}$: 
\begin{equation}
\mathscr{P}_{n}=i^{m}\{I,X,Y,Z\}^{\otimes n},\; m=0,\ldots,3\:,\label{eq:PauliGroup}
\end{equation}
where $X$, $Y$, and $Z$ are the usual Pauli matrices and $I$ is
the identity matrix. It is customary to map the Pauli operators, up
to a phase, to two binary strings, $\mathbf{v},\mathbf{u}\in\{0,1\}^{\otimes n}$
\cite{Calderbank-1997}, 
\begin{equation}
U\equiv i^{m'}X^{\mathbf{v}}Z^{\mathbf{u}}\:\rightarrow(\mathbf{v},\mathbf{u}),\label{eq:mapping}
\end{equation}
where $X^{\mathbf{v}}=X_{1}^{v_{1}}X_{2}^{v_{2}}\ldots X_{n}^{v_{n}}$
and $Z^{\mathbf{u}}=Z_{1}^{u_{1}}Z_{2}^{u_{2}}\ldots Z_{n}^{u_{n}}$.
A product of two quantum operators corresponds to a sum ($\mod2$)
of the corresponding pairs $(\mathbf{v}_{i},\mathbf{u}_{i})$.

An $[[n,k,d]]$ stabilizer code $\mathcal{Q}$ is a $2^{k}$-dimensional
subspace of the Hilbert space $\Hn$ stabilized by an Abelian stabilizer
group $\mathscr{S}$ generated by the commuting Pauli operators (generators)
$G_{1},\ldots,G_{n-k}$, i.e., $\mathscr{S}=\left\langle G_{1},\ldots,G_{n-k}\right\rangle $
and $-\openone\not\in\mathscr{S}$ \cite{gottesman-thesis}. Explicitly,
\begin{equation}
\mathcal{Q}=\{\ket\psi:S\ket\psi=\ket\psi,\forall S\in\mathscr{S}\}.\label{eq:stabilizer-code}
\end{equation}
Each generator $G_{i}\in\mathscr{S}$ is mapped according to Eq.~(\ref{eq:mapping})
in order to obtain the binary generator matrix $H=(A_{X}|A_{Z})$
in which each row corresponds to a generator, with rows of $A_{X}$
formed by $\mathbf{v}$ and rows of $A_{Z}$ formed by $\mathbf{u}$
vectors. For generality, we assume that the matrix $H$ may also contain
unimportant linearly dependent rows which are added after the mapping
has been done. The commutativity of stabilizer generators corresponds
to the following condition on the binary matrices $A_{X}$ and $A_{Z}$:
\begin{equation}
A_{X}A_{Z}^{T}+A_{Z}A_{X}^{T}=0\;(\mod2).\label{eq:product}
\end{equation}

A more narrow set of Calderbank-Shor-Steane (CSS) codes
\cite{Calderbank-Shor-1996} contains codes whose stabilizer generators
can be chosen to contain products of only Pauli $X$ or Pauli $Z$
operators. For these codes the stabilizer generator matrix can be
chosen in the form:
\begin{equation}
H=\left(\begin{array}{c|c}
G_{X} & 0\\
0 & G_{Z}
\end{array}\right),\label{eq:CSS}
\end{equation}
where the commutativity condition simplifies to $G_{X}G_{Z}^{T}=0$.

The dimension of a quantum code is 
\begin{equation}
k=n-\rank H;\label{eq:rank-dimension}
\end{equation}
for a CSS code this simplifies to 
\begin{equation}
k=n-\rank G_{X}-\rank G_{Z}.\label{eq:rank-dimension-CSS}
\end{equation}

The distance $d$ of a quantum stabilizer code is given by the minimum
weight of an operator $U$ which commutes with all operators from
the stabilizer $\mathscr{S}$, but is not a part of the stabilizer,
$U\not\in\mathscr{S}$. In terms of the binary vector pairs $(\mathbf{a},\mathbf{b})$,
this is equivalent to a minimum weight of the bitwise $\OR(\mathbf{a},\mathbf{b})$
of all pairs satisfying the symplectic orthogonality condition, 
\begin{equation}
A_{X}\mathbf{b}+A_{Z}\mathbf{a}=0,\label{eq:symplectic}
\end{equation}
which are not linear combinations of the rows of $H$. A code of distance
$d$ can detect any error of weight up to $d-1$, and correct up to
$\lfloor d/2\rfloor$.

In an equivalent representation, one can map any Pauli operator $U$ in
Eq.~(\ref{eq:mapping}), to a quaternary vector over $\mathbb{F}_{4}$,
$\mathbf{e}\equiv\mathbf{u}+\omega\mathbf{v}$. A product of two
quantum operators corresponds to a sum ($\mod2$) of the corresponding
vectors.  Two Pauli operators commute if and only if the \emph{trace
  inner product}
$\mathbf{e}_{1}*\mathbf{e}_{2}\equiv\mathbf{e}_{1}\cdot\overline{\mathbf{e}}_{2}+\overline{\mathbf{e}}_{1}\cdot\mathbf{e}_{2}$
of the corresponding vectors is zero (which is equivalent to the
symplectic orthogonality condition), where
$\overline{\mathbf{e}}\equiv\mathbf{u}+\overline{\omega}\mathbf{v}$.
With this map, generators of a stabilizer group are mapped to rows of
a generator $\mathbb{G}$ of an additive code over $\mathbb{F}_{4}$,
with the condition that the trace inner product of any two rows
vanishes\cite{Calderbank-1997} [see Eq.~(\ref{eq:product})]. The
vectors generated by rows of $\mathbb{G}$ correspond to stabilizer
generators which act trivially on the code; these vectors form the
\emph{degeneracy group} and are omitted from the distance
calculation. For CSS codes in Eq.~(\ref{eq:CSS}) the generator matrix
is a direct sum $\mathbb{G}=G_{x}\oplus\omega G_{z}$.  In the
following, we will use both, quaternary and binary, representations.

A classical LDPC code 
is a code with a sparse parity-check matrix $H$.  For a regular
$(j,l)$ quantum LDPC code, every column and every row of $H$ have
weights $j$ and $l$ respectively, while for a $(j,l)$-limited quantum LDPC
code these weights are limited from above by $j$ and $l$.  

In the case of quantum LDPC codes, these properties apply to the
generator matrix $\mathbb{G}$ whose rows correspond to the generators
of the stabilizer group.

\subsection{Bicycle codes}

In one of the first studies of quantum LDPC codes, MacKay et.\ al.\ proposed
a CSS code construction\cite{MacKay:IEEE2004} which can be written
in a block form as: 
\begin{equation}
G_{X}=G_{Z}=\left(A,A^{T}\right),\label{eq:Bicycle0}
\end{equation}
where $A$ is a binary circulant matrix. Bicycle codes are
obtained after some of the rows in $G_{X}$ or $G_{Z}$ are
deleted. Numerically, such codes show good error-correction
capabilities\cite{MacKay:IEEE2004,Poulin-Chung-2008}; however, the
distance of such codes is unknown.

\begin{figure}[htbp]
\centering \includegraphics[width=0.48\columnwidth]{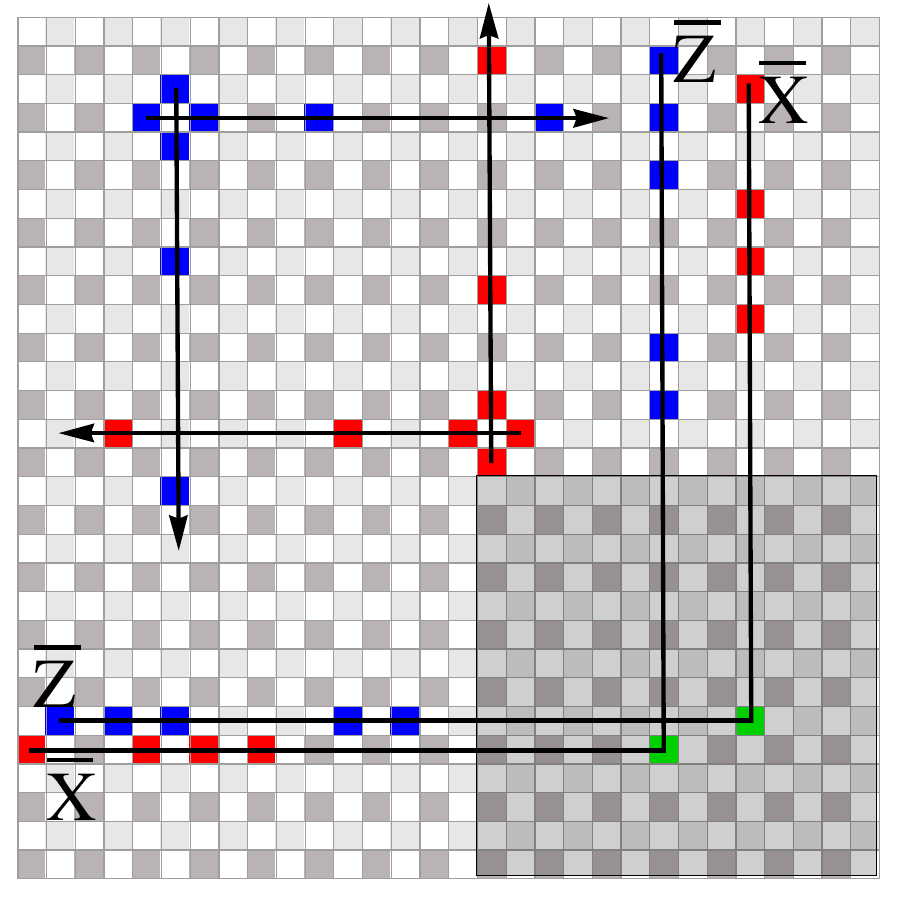}\hskip0.1in
\includegraphics[width=0.48\columnwidth]{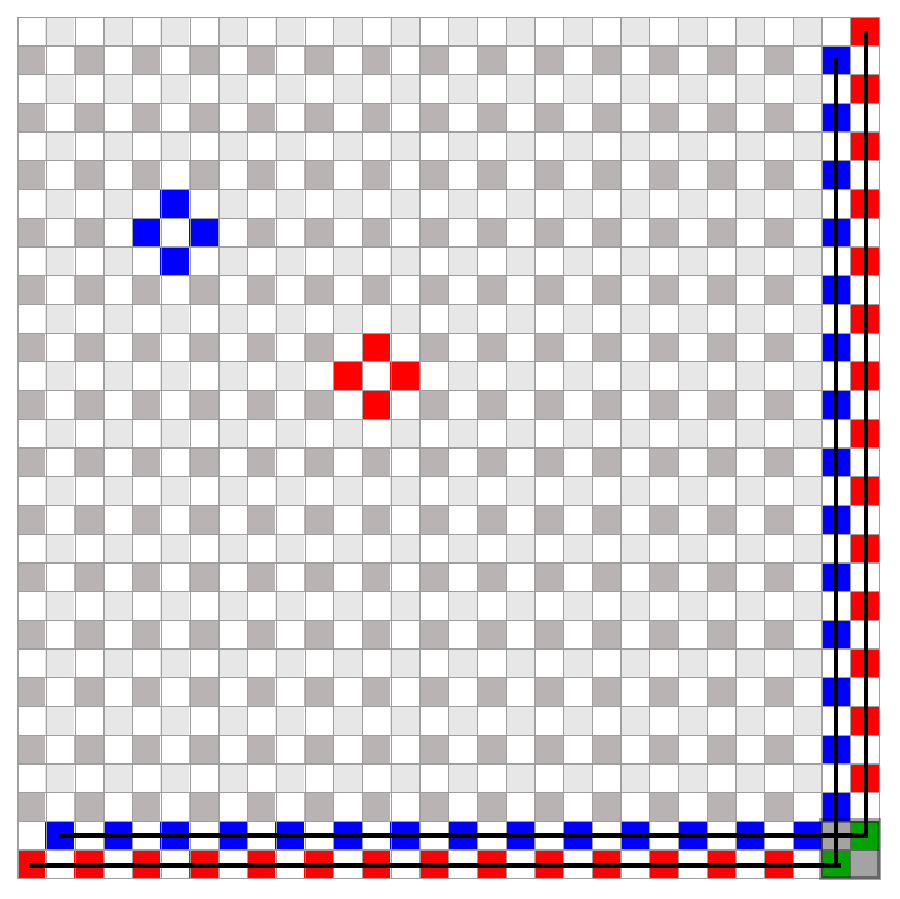} \caption{(Color online) Left: Two stabilizer generators (marked by arrows)
and two pairs of anticommuting logical operators (marked by lines)
of a $[[450,98,5]]$ code in Eq.~(\ref{eq:Till}) formed by circulant
matrices $\mathcal{H}_{1}=\mathcal{H}_{2}$ corresponding to coefficients
of the polynomial $h(x)=1+x+x^{3}+x^{7}$ (red -- $X$ operators, blue
-- $Z$ operators, green -- overlap of $Z$ and $X$ operators, dark
and light gray -- dual sublattices of physical qubits). All other
stabilizer generators are obtained by shifts over the same sublattice
with periodic boundaries. In the shaded region, each gray square uniquely
corresponds to a different logical operator, thus $98$ encoded logical
qubits. Right: same for the toric code $[[450,2,15]]$.}

\label{fig:Visualization} 
\end{figure}

\subsection{Hypergraph-product codes}

Tillich and Z\'emor proposed a CSS construction which can be interpreted
as a finite-rate generalization of toric codes\cite{Tillich2009}
and allows for LDPC constructions. For such codes, the generator matrix
is constructed from a product of two hypergraphs, each corresponding
to a parity check matrix of a classical binary code. The associated
CSS code can be recast in a matrix form with the generators given
by\cite{Kovalev:ISIT2012} 
\begin{equation}
\begin{array}{c}
{\displaystyle G_{X}=(E_{2}\otimes\mathcal{H}_{1},\mathcal{H}_{2}\otimes E_{1}),}\\
{\displaystyle G_{Z}=(\mathcal{H}_{2}^{T}\otimes\widetilde{E}_{1},\widetilde{E}_{2}\otimes\mathcal{H}_{1}^{T}).}
\end{array}\label{eq:Till}
\end{equation}
Here each sublattice block is constructed as a Kronecker product
(denoted with ``$\otimes$'') of two binary matrices $\mat{H}_{1}$
(dimensions $r_{1}\times n_{1}$) and $\mat{H}_{2}$ (dimensions
$r_{2}\times n_{2}$), and $E_{i}$ and $\widetilde{E}_{i}$, $i=1,2$,
are unit matrices of dimensions given by $r_{i}$ and $n_{i}$,
respectively. The matrices $G_{X}$ and $G_{Z}$, respectively, have
$r_{1}r_{2}$ and $n_{1}n_{2}$ rows (not all of the rows are
necessarily linearly independent), and they both have $N\equiv
r_{2}n_{1}+r_{1}n_{2}$ columns, which gives the block length of the
quantum code. The commutativity condition $G_{X}G_{Z}^{T}=0$ is
obviously satisfied by Eq.~(\ref{eq:Till}) since the Kronecker product
obeys $(A\otimes B)(C\otimes D)=AC\otimes BD$.

The parameters $[[N,K,D]]$ of thus constructed quantum code are determined
by those of the four classical codes which use the matrices $\mat{H}_{1}$,
$\mat{H}_{2}$, $\mat{H}_{1}^{T}$, and $\mat{H}_{2}^{T}$ as the
parity-check matrices. The corresponding parameters are introduced
as 
\begin{equation}
\mat{C}_{\mat{H}_{i}}=[n_{i},k_{i},d_{i}],\quad\mat{C}_{\mat{H}_{i}^{T}}=[{\widetilde{n}}_{i},\widetilde{k}_{i},\widetilde{d}_{i}],\quad i=1,2,\label{eq:params}
\end{equation}
where we use the convention \cite{Tillich2009} that the distance
$d_{i}(\widetilde{d}_{i})=\infty$ if $k_{i}(\widetilde{k}_{i})=0$.
The matrices $\mat{H}_{i}$ are arbitrary, and are allowed to have
linearly-dependent rows and/or columns. As a result, both $k_{i}=n_{i}-\rank\mat{H}_{i}$
and $\widetilde{k}_{i}=\widetilde{n}_{i}-\rank\mat{H}_{i}$ may be
non-zero at the same time as the block length of the ``transposed''
code $\mat{C}_{\mat{H}_{i}^{T}}$ is given by the number of rows of
$\mat{H}_{i}$, $\widetilde{n}_{i}=r_{i}$.

Specifically, for the hypergraph-product code (\ref{eq:Till}), we
have $N=r_{2}n_{1}+r_{1}n_{2}$, $K=2k_{1}k_{2}-k_{1}s_{2}-k_{2}s_{1}$
with $s_{i}=n_{i}-r_{i},\, i=1,2\,$ (Theorem 7 from Ref.~\cite{Tillich2009}),
while the distance $D$ satisfies the conditions $D\ge\min(d_{1},d_{2},\widetilde{d}_{1},\widetilde{d}_{2})$
(Theorem 9 from Ref.~\cite{Tillich2009}), and two upper bounds (Lemma
10 from Ref.~\cite{Tillich2009}): if $k_{1}>0$ and $\widetilde{k}_{2}>0$,
then $D\le d_{1}$; if $k_{2}>0$ and $\widetilde{k}_{1}>0$, then
$D\le d_{2}$.

A full-rank parity check matrix $\mathcal{H}_{1}$ of a binary code
with parameters $\mathcal{C}_{\mat{H}_{1}}=[n_{1},k_{1},d_{1}]$
($r_{1}=n_{1}-k_{1}$, $\widetilde{k}_{1}=0$) and
$\mathcal{H}_{2}=\mathcal{H}_{1}^{T}$ define a quantum code with
parameters
$[[(n_{1}-k_{1})^{2}+n_{1}^{2},k_{1}^{2},d_{1}]]$\cite{Tillich2009}.
Furthermore, a family of finite-rate $(h,v)$-limited classical LDPC
codes with asymptotically finite relative distance will correspond to
a family of finite-rate $(v,h+v)$-limited quantum LDPC codes with the
distance scaling as $D\propto\sqrt{N}$.

\section{Two-sublattice codes}

The commutativity condition for QECCs in Eq.~(\ref{eq:product})
puts a strong limitation on suitable parity check matrices. The problem
becomes even more difficult when the additional requirement of LDPC
structure is imposed. In particular, this strongly limits possible
counting arguments for establishing bounds on code parameters. In
such a setting, constructions based on some ansatz become very useful.
In the following, we study several CSS constructions based on two
sublattices corresponding to the columns of the binary matrices $A_{1}(B_{2}^{T})$
and $B_{1}(A_{2}^{T})$: 
\begin{equation}
G_{X}=\left(A_{1},B_{1}\right),\quad G_{Z}=\left(B_{2}^{T},A_{2}^{T}\right),\label{eq:two-sublattice-general}
\end{equation}
where the matrices $A_{i}$, $B_{i}$, $i=1,2$, satisfy the condition
$A_{1}B_{2}+B_{1}A_{2}=0$ (we assume binary linear algebra throughout
this paper).

\subsection{Two-sublattice CSS code from a generic stabilizer code}

A large number of two-sublattice CSS codes~(\ref{eq:two-sublattice-general})
can be obtained from regular stabilizer codes by the following \begin{theorem}
\label{th:two-sublattice-theorem} For any quantum stabilizer code
$[[N,K,D]]$ with the generator matrix 
\begin{equation}
H=(A|B),\label{eq:two-sublattice}
\end{equation}
there is a reversible mapping to a two-sublattice quantum CSS code~(\ref{eq:two-sublattice-general})
with $A_{1}=A_{2}^{T}=A$, $B_{1}=B_{2}^{T}=B$ and the parameters
$[[2N,2K,D']]$, where $D\le D'\le2D$. 
\end{theorem} \begin{proof} Explicitly, the generator matrices are
\begin{equation}
G_{X}=(A,B),\quad G_{Z}=(B,A).\label{eq:two-sublattice-CSS}
\end{equation}
The dimension of the code simply follows from Eqs.~(\ref{eq:rank-dimension}),
(\ref{eq:rank-dimension-CSS}), given that $\rank G_{X}=\rank G_{Z}=\rank H$.
Any binary vector $\mathbf{e}=(\mathbf{a}|\mathbf{b})$ such that
$A\mathbf{b}+B\mathbf{a}=0$ maps to a pair of double-size vectors
$\mathbf{e}_{z}=(\mathbf{b},\mathbf{a})$, $\mathbf{e}_{x}=(\mathbf{a},\mathbf{b})$
which satisfy $G_{X}\mathbf{e}_{z}=0$, $G_{Z}\mathbf{e}_{x}=0$;
the corresponding weights obey the inequality $\wgt\OR(\mathbf{a},\mathbf{b})\le\wgt(\mathbf{a},\mathbf{b})\le2\wgt\OR(\mathbf{a},\mathbf{b})$,
which ensures the conditions on the distance. 
 It is easy to check that the reverse mapping also works. \end{proof}
Note that an original code that exceeds the generic quantum Gilbert-Varshamov
(GV) bound\cite{Feng:dec.2004} is mapped to a CSS code that exceeds
the version of the GV bound specific for such codes\cite{Calderbank-Shor-1996}.
Also, an original sparse code is mapped to a sparse code, with the
same limit on the column weight, and row weight at most doubled. So,
if one has a non-CSS code and wants to use one of the measurement
techniques designed for such codes, this
can be done by first constructing the corresponding CSS code.

We use the reverse version of this mapping in Sec.~\ref{sec:non-CSS-hb}
to construct the non-CSS versions of hyperbicycle codes.

\subsection{Generalized bicycle codes}

Let us now start with two commuting square $n$ by $n$ binary matrices,
$AB+BA=0$. Then, we can satisfy the general two-sublattice ansatz~(\ref{eq:two-sublattice-general})
by taking $A_{1}=A_{2}=A$, $B_{1}=B_{2}=B$, which gives 
\begin{equation}
G_{X}=\left(A,B\right),\quad G_{Z}=\left(B^{T},A^{T}\right).\label{eq:Bicycle-1}
\end{equation}

In particular, the commutativity is guaranteed for circulant matrices
$A$ and $B$, which corresponds to a generalization of the bicycle
codes\cite{MacKay:IEEE2004}, see Eq.~(\ref{eq:Bicycle0}). In this
case, we map the linear combinations of rows in $G_{X}$ to a classical
length-$n$ additive cyclic code over $\mathbb{F}_{4}$, where elements
of the code are constructed from the generator matrix $\mathbb{G}=\omega A+B$.
If the circulant matrices $A$ and $B$ are generated by the polynomials
$f_{1}(x)$ and $f_{2}(x)$, respectively, the space of the additive
code corresponding to the stabilizer is generated by the polynomial
$g(x)=f_{1}(x)\omega+f_{2}(x)$ modulo $x^{n}-1$. That is, elements
of the stabilizer are given by the coefficients of the polynomials
$b(x)g(x)\mod(x^{n}-1)$, with arbitrary binary $b(x)$.

A canonical form of cyclic additive codes over $\mathbb{F}_{4}$ has
been introduced in Ref.~\onlinecite{Calderbank-1997}, where Theorem
14 states that any cyclic additive code can be represented via two
generators as $\left\langle \omega p(x)+q(x),r(x)\right\rangle $
$\mod(x^{n}-1)$ with $p(x)=\gcd[f_{1}(x),x^{n}-1]$, $r(x)=\gcd[(x^{n}-1)f_{2}(x)/p(x),x^{n}-1]$
and $\deg q(x)<\deg r(x)$ ($\gcd$ stands for the greatest common
divisor). In this case the code dimensionality is $k=2n-\deg p(x)-\deg r(x)$.
The special case of cyclic additive codes with a single generator
has been analyzed in Ref.~\cite{Kovalev-PRA2011} in which case the
code dimensionality simplifies to $k=n-\deg p(x)$, which formally
corresponds to $r(x)=x^{n}-1$. Note that, unlike the case of the
usual quantum additive cyclic codes\cite{Calderbank-1997,Kovalev-PRA2011},
the mapping from Eq.~(\ref{eq:Bicycle-1}) works for any circulant
matrices $A$ and $B$; no additional commutativity condition is needed
for the generator polynomials $p(x)$, $q(x)$, and $r(x)$. The parameters
of thus obtained quantum codes are given by

\begin{theorem} \label{th:tiled4} The generalized bicycle codes
in Eq.~(\ref{eq:Bicycle-1}) have the block length $N=2n$, the number
of encoded qubits $K=2\deg p(x)+2\deg r(x)-2n$ {[}$K=2\deg p(x)$
in the single generator case{]} and the distance exceeding or equal
that of the classical code over $\mathbb{F}_{4}$ formed by codewords
orthogonal to $\mathbb{G}$ with respect to the trace inner product.
\end{theorem} \begin{proof} The code dimensionality immediately
follows from the parameters of the canonical form of the code generated
by $g(x)$. The orthogonal code contains the quantum code, hence the
distance estimate. \end{proof} Note that the distance estimate in
Theorem \ref{th:tiled4} is tight only for pure codes since a possibility
for degeneracy is not taken into consideration.

\begin{example}Suppose a cyclic linear code $[n,k,d]$ over $\mathbb{F}_{4}$
with a generator polynomial $\varrho(x)$ that divides $x^{n}-1$
generates a code space $\mathbb{G}^{\perp}$. Then the parameters
of the quantum CSS code in Eq.~(\ref{eq:Bicycle-1}) are $[[2n,2n-4\deg\varrho(x),\geq d]]$.
This construction is similar to non-CSS code construction from linear
cyclic codes in Ref. \cite{Calderbank-1997}, except that here the
dual code does not have to be self-orthogonal. For a cyclic $[30,25,4]$
code with $\varrho(x)=(1+x)^{2}(1+\omega x)(1+x+\omega x^{2})$ we
obtain a quantum code $[[60,40,4]]$. \end{example}

\begin{example}A CSS family of odd-distance rotated toric codes \cite{Kovalev:ISIT2012}
is obtained for $f_{1}(x)=(1+x^{2t^{2}+1})$ and $f_{2}(x)=x(1+x^{2t^{2}-1})$,
$t=1,2,\ldots$ by construction in Eq.~(\ref{eq:Bicycle-1}). These
codes have the parameters $[[2t^{2}+2(t+1)^{2},2,2t+1]]$. Explicitly,
$[[10,2,3]]$, $[[26,2,5]]$, $[[50,2,7]]$, $[[82,2,9]]$, \ldots{}.
\end{example}

The constructions in Theorems~\ref{th:two-sublattice-theorem} and
\ref{th:tiled4} {[}Eqs.~(\ref{eq:two-sublattice-CSS}) and (\ref{eq:Bicycle-1})
respectively{]} coincide for symmetric matrices, $A=A^{T}$, $B=B^{T}$.
Then, starting with two palindromic polynomials $f_{i}(x)=x^{\deg f_{i}(x)}f_{i}(1/x)$,
$i=1,2$, we can obtain non-CSS halved hyperbicycle codes in Eq.~(\ref{eq:two-sublattice})
by applying the reverse of Theorem \ref{th:two-sublattice-theorem}
to the matrices $A$ and $B$ \cite{Kovalev:ISIT2012}.

\begin{example}A non-CSS family of smallest odd-distance rotated
toric codes \cite{Kovalev-PRA2011} is obtained for palindromic $f_{1}(x)=x^{t}(1+x^{2t^{2}+1})$
and $f_{2}(x)=x^{t+1}(1+x^{2t^{2}-1})$, $t=1,2,\ldots$ by construction
in Eq.~(\ref{eq:two-sublattice}). These codes have the parameters
$[[t^{2}+(t+1)^{2},1,2t+1]]$. Explicitly, $[[5,1,3]]$, $[[13,1,5]]$,
$[[25,1,7]]$, $[[41,1,9]]$, \ldots{}. \end{example} The codes
in the last two examples exceed the lower bound in Theorem \ref{th:tiled4}
due to degeneracy.

\begin{table*}[t]
\begin{tabular}{|c|c|}
\hline 
Code $1$.  & $\begin{array}{l}
A=\mathcal{H}_{1}\otimes E\otimes E+E\otimes\mathcal{H}_{1}\otimes E+E\otimes E\otimes\mathcal{H}_{1}\\
B=\mathcal{H}_{1}\otimes\mathcal{H}_{1}\otimes E+E\otimes\mathcal{H}_{1}\otimes\mathcal{H}_{1}+\mathcal{H}_{1}\otimes E\otimes\mathcal{H}_{1}
\end{array}$\tabularnewline
\hline 
\hline 
Code $2$.  & $\begin{array}{l}
A=\mathcal{H}_{1}\otimes E\otimes E+E\otimes\mathcal{H}_{1}\otimes E+E\otimes\mathcal{H}_{1}\otimes\mathcal{H}_{1}+\mathcal{H}_{1}\otimes E\otimes\mathcal{H}_{1}+\mathcal{H}_{1}\otimes\mathcal{H}_{1}\otimes\mathcal{H}_{1}\\
B=E\otimes E\otimes\mathcal{H}_{1}+\mathcal{H}_{1}\otimes\mathcal{H}_{1}\otimes E+E\otimes\mathcal{H}_{1}\otimes\mathcal{H}_{1}+\mathcal{H}_{1}\otimes E\otimes\mathcal{H}_{1}
\end{array}$\tabularnewline
\hline 
Code $3$.  & $\begin{array}{l}
A=\mathcal{H}_{1}\otimes E\otimes E+E\otimes\mathcal{H}_{1}\otimes E+E\otimes\mathcal{H}_{1}\otimes\mathcal{H}_{1}+\mathcal{H}_{1}\otimes E\otimes\mathcal{H}_{1}+\mathcal{H}_{1}\otimes\mathcal{H}_{1}\otimes\mathcal{H}_{1}\\
B=\mathcal{H}_{1}\otimes E\otimes E+E\otimes\mathcal{H}_{1}\otimes E+E\otimes E\otimes\mathcal{H}_{1}+\mathcal{H}_{1}\otimes\mathcal{H}_{1}\otimes E+E\otimes\mathcal{H}_{1}\otimes\mathcal{H}_{1}+\mathcal{H}_{1}\otimes\mathcal{H}_{1}\otimes\mathcal{H}_{1}
\end{array}$\tabularnewline
\hline 
Code $4$.  & $\begin{array}{l}
A=E\otimes\mathcal{H}_{1}\otimes\mathcal{H}_{1}+\mathcal{H}_{1}\otimes E\otimes\mathcal{H}_{1}+E\otimes\mathcal{H}_{1}\otimes E\\
B=\mathcal{H}_{1}\otimes E\otimes E+E\otimes E\otimes\mathcal{H}_{1}+E\otimes\mathcal{H}_{1}\otimes\mathcal{H}_{1}+\mathcal{H}_{1}\otimes\mathcal{H}_{1}\otimes\mathcal{H}_{1}
\end{array}$\tabularnewline
\hline 
\end{tabular}\caption{Tensor-product two-sublattice representation of Haah's codes corresponding
to Eq.~(\ref{eq:two-sublattice-general}) where $\mathcal{H}_{1}$
is a circulant matrix corresponding to parity check polynomial $p(x)=1+x$
of a repetition code, $E$ is a unit matrix of the same dimensions
with $\mathcal{H}_{1}$ and the summation is $\mod2$.}

\label{tab:table} 
\end{table*}

\subsection{Tensor-product constructions and Haah's codes}

Further generalization of the bicycle-like construction in Eq.~(\ref{eq:Bicycle-1})
can be achieved by combining tensor products with commuting (e.g.,
circulant) matrices. The most general form of two-sublattice tensor-product
codes has the form: 
\[
\begin{array}{c}
A=\sum_{i_{1}\ldots i_{k}}\mathcal{H}_{i_{1}}^{A}\otimes\ldots\otimes\mathcal{H}_{i_{k}}^{A},\\
B=\sum_{i_{1}\ldots i_{k}}\mathcal{H}_{i_{1}}^{B}\otimes\ldots\otimes\mathcal{H}_{i_{k}}^{B},
\end{array}
\]
where $\mathcal{H}_{i_{l}}^{A}$ and $\mathcal{H}_{i_{l}}^{B}$ are
matching, pairwise-commuting binary square matrices
(i.e.,
$\mathcal{H}_{i_{l}}^{A}\mathcal{H}_{j_{l}}^{B}
+\mathcal{H}_{j_{l}}^{B}\mathcal{H}_{i_{l}}^{A}=0$). For circulant
matrices $\mathcal{H}_{i_{l}}^{A}$ and 
$\mathcal{H}_{j_{l}}^{B}$ the commutativity is automatically
satisfied.

Several examples of such codes are given by the Haah's
codes\cite{Haah:PRA2011}, which give examples of local codes in 3D
without string logical operators, and may lead to realizations of
self-correcting quantum memories.  Such codes are defined on two
sublattices and have exactly the tensor-product structure discussed
here. In Table \ref{tab:table}, we list four codes presented in
Ref.~\cite{Haah:PRA2011}. These codes are essentially constructed from
a repetition code and it is straightforward to generalize this
construction to arbitrary cyclic binary codes by using the
corresponding binary circulant matrix $\mathcal{H}_{1}$. The
commutativity of matrices $A$ and $B$ in
Eq.~(\ref{eq:two-sublattice-general}) immediately follows.

A non-CSS generalization of construction in Table \ref{tab:table}
can be achieved by using symmetric circulant matrices {[}see non-CSS
construction in Eq.~(\ref{eq:two-sublattice}){]}.

\section{CSS and non-CSS hyperbicycle codes}

This section contains our most important results. We show that the
families of hypergraph-product and generalized bicycle codes can be
obtained as limiting cases of a larger family of \emph{hyperbicycle}
codes. The main advantage of this construction is that it gives a
number of previously unreported families of quantum codes with tight
bounds on, or even explicitly known distance. This includes many strongly
degenerate LDPC codes with the distance much greater than the maximum
weight of a stabilizer generator. In this section we discuss the construction
of such codes, their parameters, and give examples.

\subsection{CSS hyperbicycle codes: construction}

We define the hyperbicycle CSS codes as follows: 
\begin{equation}
\begin{array}{c}
{\displaystyle G_{X}={\displaystyle \Bigl(E_{b}\otimes{\textstyle \sum}_{i}I_{i}^{(\chi)}\otimes a_{i},\;{\textstyle \sum}_{i}b_{i}\otimes I_{i}^{(\chi)}\otimes E_{a}\Bigr)},}\\
{\displaystyle G_{Z}=\left({\textstyle \sum}_{i}b_{i}^{T}\otimes\widetilde{I}_{i}^{(\chi)}\otimes\widetilde{E}_{a},\;\widetilde{E}_{b}\otimes{\textstyle \sum_{i}}\widetilde{I}_{i}^{(\chi)}\otimes a_{i}^{T}\right).}
\end{array}\label{eq:Hyperbicycle}
\end{equation}
Here we introduce two sets of binary matrices $a_{i}$ (dimensions
$r_{1}\times n_{1}$, $i=0,\ldots,c-1$) and $b_{i}$ (dimensions
$r_{2}\times n_{2}$, $i=0,\ldots,c-1$); $E_{a}$, $E_{b}$, $\widetilde{E}_{a}$
and $\widetilde{E}_{b}$ are unit matrices of dimensions given by
$r_{1}$, $r_{2}$, $n_{1}$ and $n_{2}$, respectively. Matrices
$I_{i}^{(\chi)}$ ($\widetilde{I}_{i}^{(\chi)}$) are permutation
matrices given by a product of two permutation matrices, i.e. $I_{i}^{(\chi)}=S_{\chi}I_{i}$
($\widetilde{I}_{i}^{(\chi)}=S_{\chi}^{T}I_{i}^{T}$) where $(I_{i})_{kj}=\delta_{j-k,i\mod c}$
is a circulant permutation matrix, $(S_{\chi})_{kj}=\delta_{j-k,(k-1)(\chi-1)\mod c}$
and the positive integers $c$ and $\chi$ are coprime . A version
of this construction for $c=2$ and $\chi=1$ has been previously
reported by us in Ref.~\onlinecite{Kovalev:ISIT2012}.

The matrices $G_{X}$ and $G_{Z}$, respectively, have $cr_{1}r_{2}$
and $cn_{1}n_{2}$ rows (not all of the rows are linearly independent),
and they both have 
\begin{equation}
N\equiv c(r_{1}n_{2}+r_{2}n_{1})\label{eq:block-size}
\end{equation}
columns, which gives the block length of the quantum code. The commutativity
condition $G_{X}G_{Z}^{T}=0$ is obviously satisfied by Eq.~(\ref{eq:Hyperbicycle})
since the permutation matrices commute with each other. Note that
for $c=1$ and $\chi=1$ we recover the hypergraph-product codes in
Eq.~(\ref{eq:Till}) and for $r_{i}=n_{i}=1$, $i=1,2$, (i.e., $a_{i}$
and $b_{i}$ given by binary numbers) we recover the generalized bicycle
code construction in Eq.~(\ref{eq:Bicycle-1}).

In order to characterize codes in Eq.~(\ref{eq:Hyperbicycle}) it
is convenient to introduce the ``tiled'' binary matrices: 
\begin{equation}
\begin{array}{c}
\mathcal{H}_{1}={\textstyle \sum}_{i}I_{i}^{(\chi)}\otimes a_{i},\quad\mathcal{H}_{2}={\textstyle \sum}_{i}b_{i}\otimes I_{i}^{(\chi)},\\
\widetilde{\mathcal{H}}_{1}={\textstyle \sum}_{i}\widetilde{I}_{i}^{(\chi)}\otimes a_{i}^{T},\quad\widetilde{\mathcal{H}}_{2}={\textstyle \sum}_{i}b_{i}^{T}\otimes\widetilde{I}_{i}^{(\chi)}.
\end{array}\label{eq:check-tiled}
\end{equation}
For example, for $c=5$ and $\chi=2$ we have
\begin{equation}
\mathcal{H}_{1}=\left(\begin{array}{ccccc}
a_{1} & a_{2} & a_{3} & a_{4} & a_{5}\\
a_{4} & a_{5} & a_{1} & a_{2} & a_{3}\\
a_{2} & a_{3} & a_{4} & a_{5} & a_{1}\\
a_{5} & a_{1} & a_{2} & a_{3} & a_{4}\\
a_{3} & a_{4} & a_{5} & a_{1} & a_{2}
\end{array}\right);\label{eq:form}
\end{equation}
note that the subsequent block rows are shifted by $\chi=2$ positions.

For the following discussion it is also useful to define auxiliary
binary matrices:
\begin{equation}
\begin{array}{c}
\mathcal{H}_{1}^{0}={\textstyle \sum}_{i}I_{i}\otimes a_{i},\quad\mathcal{H}_{2}^{0}={\textstyle \sum}_{i}b_{i}\otimes I_{i},\end{array}\label{eq:check-tiled-1}
\end{equation}
where in terms of these matrices we can write:
\begin{equation}
\begin{array}{c}
\mathcal{H}_{1}=S_{\chi}\otimes E_{a}\cdot\mathcal{H}_{1}^{0},\quad\mathcal{H}_{2}=E_{b}\otimes S_{\chi}\cdot\mathcal{H}_{2}^{0},\\
\widetilde{\mathcal{H}}_{1}=S_{\chi}^{T}\otimes E_{a}\cdot\mathcal{H}_{1}^{0T},\quad\widetilde{\mathcal{H}}_{2}=E_{b}\otimes S_{\chi}^{T}\cdot\mathcal{H}_{2}^{0T}.
\end{array}\label{eq:check-tiled-2}
\end{equation}
With this notation, it is clear that the generator matrices (\ref{eq:Hyperbicycle})
correspond to the hypergraph generators (\ref{eq:Till}), except that
the identity matrices $E_{i}$, $\widetilde{E}_{i}$ are now reduced
in size by a factor of $1/c$. The block length of the original hypergraph
code with the present binary matrices~(\ref{eq:check-tiled}) is
\begin{equation}
N_{\mathrm{orig}}=c^{2}(r_{1}n_{2}+r_{2}n_{1}).\label{eq:N-orig}
\end{equation}
In this sense, one can view the codes defined by Eq.~(\ref{eq:Hyperbicycle})
with $c>1$ as reduced hypergraph codes.

\subsection{CSS hyperbicycle codes: dimension}

Just as for the hypergraph codes, the parameters of classical codes
$\mathcal{C}_{\mathcal{H}_{i}}$ and $\mathcal{C}_{\widetilde{\mathcal{H}}_{i}}$
with parity check matrices $\mathcal{H}_{i}$ and $\widetilde{\mathcal{H}}_{i}$,
respectively, contain information about the parameters of the quantum
code in Eq.~(\ref{eq:Hyperbicycle}). We denote the distances of
these binary codes as $d_{i}$ and $\widetilde{d}_{i}$, and their
dimensions as $k_{i}$ and $\widetilde{k}_{i}$, $i=1,2$.

Regardless of the choice of the matrices $a_{i}$, $b_{i}$, the codes
$\mathcal{C}_{\mathcal{H}_{i}}$ and $\mathcal{C}_{\widetilde{\mathcal{H}}_{i}}$
are quasicyclic, with the cycle of length equal the dimension $c$
of the cyclic permutation matrices $I_{\chi i}$ ($\widetilde{I}_{\chi i}$),
$i=1,2$. Indeed, the corresponding block shifts merely lead to permutations
of rows of the check matrices $\mathcal{H}_{i}$,
$\widetilde{\mathcal{H}}_{i}$ [see Eq.~(\ref{eq:form})].
In order to define the dimension of the corresponding hyperbicycle
codes with generators~(\ref{eq:Hyperbicycle}), we first classify
the vectors in $\mathcal{C}_{\mathcal{H}_{i}}$ and $\mathcal{C}_{\widetilde{\mathcal{H}}_{i}}$
with respect to this circulant symmetry.

We start with the case of a binary cyclic code with block length $c$,
with the generator polynomial $g(x)$, which divides $x^{c}-1$. {[}Polynomial
algebra in this section is done modulo 2.{]} Any codeword corresponds
to a polynomial $w(x)$ which contains $g(x)$ as a factor, and, therefore,
every cyclotomic root of $g(x)$ is also a root of $w(x)$. However,
the particular polynomial $w(x)=g(x)f(x)$ may also contain other
factors of $x^{c}-1$ and thus have symmetry different from that of
$g(x)$. We can define a linear space ${\cal C}^{(p)}$ of length-$c$
vectors corresponding to $w(x)$ with the exact symmetry of $g(x)$
where $p(x)\equiv(x^{c}-1)/g(x)$ by defining the equivalence $w_{1}(x)\equiv w_{2}(x)$
for a given $g(x)$ as $f_{1}=f_{2}\mod p'(x)$ for all $p'(x)$ such
that $p'(x)\neq p(x)$ is a factor of $p(x)$. The same equivalence
can be also defined modulo greatest common divisor (gcd) of all such
polynomials $p'(x)$. In terms of the corresponding check polynomial
$p(x)$, the dimension $k_{0}^{(p)}$ of thus defined space ${\cal C}^{(p)}$
is zero unless $p(x)$ is a non-zero power of an irreducible polynomial
$p_{\alpha}(x)$, in which case $k_{0}^{(p)}=\deg p_{\alpha}(x)$.

For the quasicyclic code $\mathcal{C}_{\mathcal{H}_{1}}$ with the
first check matrix in Eq.~(\ref{eq:check-tiled}), the vector $\mathbf{w}$
is in the symmetry class of $p(x)$, where $p(x)$ divides $x^{c}-1$,
if $\mathbf{w}$ satisfies the condition $[p(I_{1})\otimes E_{a}]\,\mathbf{w}=0$
and is not a member of such a symmetry class of any factor of $p(x)$.
For each polynomial $p(x)$, the l.h.s.\ in these equations is a
sum of cyclic shifts of the vector $\mathbf{w}$ corresponding to
each non-zero coefficient of $p(x)$. We denote the dimension of the
subcode of $\mathcal{C}_{\mathcal{H}_{1}}$ with all vectors in the
symmetry class of $p(x)$ as $k_{1}^{(p)}$. The symmetry implies
that $k_{1}^{(p)}$ must contain the dimension $k_{0}^{(p)}$ introduced
in the previous paragraph as a factor, and, in particular, $k_{1}^{(p)}$
must be zero whenever $k_{0}^{(p)}$ is zero. A convenient basis of
${\cal C}_{{\cal H}_{1}}^{(p)}$ can be constructed using the following
\begin{statement}\label{th:lemma-H1} Any vector of the subcode ${\cal C}_{{\cal H}_{1}}^{(p)}$
can be chosen in the form 
\begin{equation}
\mathbf{w}=\sum_{i=0}^{k_{0}^{(p)}-1}I_{i}\mathbf{g}\otimes\boldsymbol{\alpha}_{i},\label{eq:h1-expansion}
\end{equation}
where the vector $\mathbf{g}$ corresponds to the generating polynomial
$g(x)\equiv(x^{c}-1)/p(x)$; the vectors $(I_{s}\otimes E_{a})\mathbf{w}$,
$0\le s<k_{0}^{(p)}$, are linearly independent. \end{statement}
\begin{proof} Any vector of the subcode ${\cal C}_{{\cal H}_{1}}^{(p)}$
can be expanded in the form $\boldsymbol{\omega}_{i}^{(p)}\otimes\mathbf{e}_{i}$,
where $\mathbf{e}_{i}$ are all distinct weight-one vectors, and $\boldsymbol{\omega}_{i}^{(p)}$
are vectors from ${\cal C}^{(p)}$. Generally, any vector $\boldsymbol{\omega}\in{\cal C}^{(p)}$
can be written as a sum of shifts of the vector $\mathbf{g}$, $\sum_{s=0}^{k_{0}^{(p)}-1}I_{s}\mathbf{g}$;
we obtain Eq.~(\ref{eq:h1-expansion}) by rearranging the summations.
Linear independence follows from the symmetry of the vectors $\boldsymbol{\omega}_{i}^{(p)}\in{\cal C}^{(p)}$.
\end{proof}

Note that, in addition to the symmetric vectors, the code $\mathcal{C}_{\mathcal{H}_{1}}$
may contain vectors with no special symmetry with respect to the discussed
block shifts. We will formally assign these to the check polynomial
$p(x)=x^{c}-1$, and define $k_{0}^{(x^{c}-1)}\equiv1$.

For the vectors of the code $\mathcal{C}_{\mathcal{H}_{2}}$ with
the second check matrix in Eq.~(\ref{eq:check-tiled}), the condition
to be in the symmetry class of $p(x)$ reads $[E_{b}\otimes p(I_{1})]\mathbf{w}=0$
while $[E_{b}\otimes p'(I_{1})]\mathbf{w}\neq0$ for all $p'(x)\neq p(x)$
that divide $p(x)$. We denote the dimension of the corresponding
subcode as $k_{2}^{(p)}$. The same classification can be done for
codes $\mathcal{C}_{\widetilde{\mathcal{H}}_{i}}$ with the transposed
check matrices; the corresponding dimensions are ${\widetilde{k}}_{i}^{(p)}$,
$i=1,2$.


The introduced symmetry classification is in the heart of the following
\begin{statement}\label{th:lemma-0} A vector $\boldsymbol{\upsilon}$
that belongs to both $\mathcal{C}_{E_{b}\otimes\mathcal{H}_{1}}$
and $\mathcal{C}_{\mathcal{H}_{2}\otimes E_{a}}$ must be in the same
symmetry class $p(x)$ with respect to both codes ${\cal H}_{1}$
and ${\cal H}_{2}$, including the no-symmetry case $p(x)=x^{c}-1$.
Any such vector can be generally expanded in terms of 
\begin{equation}
\boldsymbol{\upsilon}_{\alpha,\beta}=\sum_{i,j=0}^{k_{0}^{(p)}-1}\boldsymbol{\beta}_{i}\otimes I_{i+j}\mathbf{g}\otimes\boldsymbol{\alpha}_{j},\label{eq:h1h2-expansion}
\end{equation}
where $\sum_{i}\boldsymbol{\beta}_{i}\otimes I_{i}\mathbf{g}\in{\cal C}_{{\cal H}_{2}}^{(p)}$
and $\sum_{i}I_{i}\mathbf{g}\otimes\boldsymbol{\alpha}_{i}\in{\cal C}_{{\cal H}_{1}}^{(p)}$
and $\mathbf{g}$ corresponds to the polynomial $g(x)\equiv(x^{c}-1)/p(x)$.
\end{statement} \begin{proof} The parameter $\chi$ does not enter
this discussion since it corresponds to permutations of rows in matrices
${\cal H}_{1}^{0}$ and ${\cal H}_{2}^{0}$. That the symmetry must
be the same becomes evident if we write the most general expansion
\begin{equation}
\boldsymbol{\upsilon}=\sum_{ij}\mathbf{e}_{i}^{2}\otimes\boldsymbol{\gamma}_{ij}\otimes\mathbf{e}_{j}^{1},\label{eq:expansion-h1h2-bare}
\end{equation}
where $\boldsymbol{\gamma}_{ij}$ are length-$c$ vectors and $\mathbf{e}_{i}^{1(2)}$
are distinct weight-1 vectors. Indeed, the condition to be in the
symmetry class of $p(x)$ is the same for both codes: $[E_{b}\otimes p(I_{1})\otimes E_{a}]\boldsymbol{\upsilon}=0$
for $p(x)$ itself but not for any of its factors; thus $\boldsymbol{\gamma}_{ij}$
must be in ${\cal C}^{(p)}$. The expansion~(\ref{eq:h1h2-expansion})
follows from Lemma~\ref{th:lemma-H1}. \end{proof}

We can now count linearly-independent rows in the generator matrices:
\begin{statement}\label{th:lemma-1} The numbers of linearly independent
rows in matrices~(\ref{eq:Hyperbicycle}) are 
\begin{eqnarray}
\rank G_{X} & = & r_{1}r_{2}c-\sum_{l}\widetilde{k}_{1}^{(p_{l})}\widetilde{k}_{2}^{(p_{l})}/k_{0}^{(p_{l})},\nonumber \\
\rank G_{Z} & = & n_{1}n_{2}c-\sum_{l}k_{1}^{(p_{l})}k_{2}^{(p_{l})}/k_{0}^{(p_{l})},\label{eq:rank-1}
\end{eqnarray}
where $p_{l}(x)$ are all binary factors of $x^{c}-1$ such that $k_{0}^{(p_{l})}\neq0$,
including $x^{c}-1$ itself. \end{statement} \begin{proof} We first
count linearly-dependent rows in $G_{Z}$. Notice that the equations
$\boldsymbol{\upsilon}^{T}\cdot(E_{b}\otimes\widetilde{\mathcal{H}}_{1})=0$
and $\boldsymbol{\upsilon}^{T}\cdot(\widetilde{\mathcal{H}}_{2}\otimes E_{a})=0$
are both satisfied for $\boldsymbol{\upsilon}$ in Eq.~(\ref{eq:h1h2-expansion})
{[}Lemma~\ref{th:lemma-0}{]}. Each pair $(\boldsymbol{\alpha},\boldsymbol{\beta})$
generates $k_{0}^{(p)}$ linearly-independent vectors, same as each
of them generates for the corresponding subcodes ${\cal C}_{{\cal H}_{i}}^{(p)}$,
$i=1,2$, respectively. Thus there are exactly $k_{1}^{(p)}k_{2}^{(p)}/k_{0}^{(p)}$
linearly-independent vectors corresponding to every $p(x)$ with non-empty
${\cal C}^{(p)}$. Such vectors have to be complemented with the pairs
of vectors of no symmetry (if any) which formally correspond to $p(x)=x^{c}-1$
and $k_{0}^{(p)}=1$. According to Lemma \ref{th:lemma-0} these are
all possible solutions, which gives $\rank G_{Z}$ in Eq.~(\ref{eq:rank-1}).
We obtain $\rank G_{X}$ by substituting the parameters of the codes
with the parity check matrices $\widetilde{\mathcal{H}}_{1}$, $\widetilde{\mathcal{H}}_{2}$.
\end{proof} We finally obtain \begin{theorem}\label{th:tiled-size}
A quantum CSS code with generators~(\ref{eq:Hyperbicycle}) encodes
\begin{equation}
K=2\sum_{l}k_{1}^{(p_{l})}k_{2}^{(p_{l})}/k_{0}^{(p_{l})}-k_{1}s_{2}-k_{2}s_{1}\label{eq:parameters}
\end{equation}
qubits, where $p_{l}(x)$ are all binary factors of $x^{c}-1$ such
that $k_{0}^{(p_{l})}\neq0$, including $x^{c}-1$ itself, and $s_{i}=n_{i}-r_{i}$,
$i=1,2$. \end{theorem} \begin{proof} The number of encoded qubits
$K$ can be deduced from Lemma \ref{th:lemma-1} using the relation
\begin{equation}
k_{i}^{(p)}-\widetilde{k}_{i}^{(p)}=s_{i}k_{0}^{(p)},\; i=1,2.\label{eq:subcode-transposed-dim}
\end{equation}
The latter follows from the fact that the rank of a matrix does not
change under transposition (and also under permutations
of rows and columns, e.g., as needed to transform $\mathcal{H}_{i}$ into
$\widetilde{\mathcal{H}}_{i}$). Specifically, restricting the action
of matrices $\mathcal{H}_{i}$ and $\widetilde{\mathcal{H}}_{i}$
to subspace $\mathcal{C}^{(p)}$, we obtain reduced mutually transposed
matrices of dimensions given by $r_{i}k_{0}^{(p)}\times n_{i}k_{0}^{(p)}$
and $n_{i}k_{0}^{(p)}\times r_{i}k_{0}^{(p)}$, which immediately
gives Eq.~(\ref{eq:subcode-transposed-dim}). \end{proof}

By construction, any $k_{i}^{(p)}$ may only be non-zero if the corresponding
$k_{i}>0$, $i=1,2$. Then, Eq.~(\ref{eq:parameters}) gives \begin{consequence}\label{th:zero-code}
A quantum CSS code with generators~(\ref{eq:Hyperbicycle}) can only
have $K>0$ if at least one of the binary codes with the parity check
matrices (\ref{eq:check-tiled}) is non-empty. \end{consequence}

\subsection{CSS hyperbicycle codes: general distance bounds}

\begin{theorem} \label{th:css-lower-bound-generic} The minimum distance
of the code with generators~(\ref{eq:Hyperbicycle}) satisfies the
lower bound 
\begin{equation}
D\geq\lfloor d/c\rfloor,\quad d\equiv\min(d_{1},d_{2},\widetilde{d}_{1},\widetilde{d}_{2}).\label{eq:css-lower-bound-generic}
\end{equation}
\end{theorem} \begin{proof} Consider a vector $\mathbf{u}$ such
that $G_{X}\cdot\mathbf{u}=0$. We construct a reduced quantum code
in the form (\ref{eq:Hyperbicycle}), with the same $c$, by keeping
only those columns of the matrices $a_{i}$, $b_{i}$ that are involved
in the product $G_{X}\cdot\mathbf{u}$. This way, for every non-zero
bit of $\mathbf{u}$, one of the reduced matrices $\mat{H}_{1}'$,
$\mat{H}_{2}'$ {[}see Eq.~(\ref{eq:check-tiled}){]} may get $c$
columns, so that these matrices have no more than $c\wgt(\mathbf{u})$
columns. If we take $\wgt(\mathbf{u})<\lfloor d/c\rfloor$, according
to Consequence~\ref{th:zero-code}, the reduced code encodes no qubits,
thus the corresponding reduced $\mathbf{u}'$, $G_{X}'\cdot\mathbf{u}'=0$,
has to be a linear combination of the rows of $G_{Z}'$. The rows
of $G_{Z}'$ are a subset of those of $G_{Z}$, with some all-zero
columns removed; thus the full vector $\mathbf{u}$ is also a linear
combination of the rows of $G_{Z}$. Similarly, a vector $\mathbf{v}$
such that $G_{Z}\cdot\mathbf{v}=0$ and $\wgt(\mathbf{v})<\left\lfloor d/c\right\rfloor $,
is a linear combination of rows of $G_{X}$. \end{proof}

Let us introduce the minimum distances $d_{i}^{(p)}$ corresponding to
the subset of the vectors of the code $\mathcal{C}_{\mathcal{H}_{i}}$
which \emph{contain} one of the vectors with the exact symmetry of
$p(x)$,
\begin{equation}
  d_{i}^{(p)}=\min\{\wgt(\mathbf{a}+\mathbf{b})|
  \mathbf{0}\neq\mathbf{a}\in\mathcal{C}_{\mathcal{H}_{i}}^{(p)},\;
  \mathbf{b}\in\mathcal{C}_{\mathcal{H}_{i}}
  \setminus\mathcal{C}_{\mathcal{H}_{i}}^{(p)}\}.\label{eq:subset-distance} 
\end{equation}
Evidently, thus introduced distances satisfy 
\begin{equation}
d_{i}^{(p)}\ge d_{i},\quad\min_{l}d_{i}^{(p_{l})}=d_{i},\;
i=1,2;\label{eq:subset-distance-limits} 
\end{equation}
the minimum is taken over all $p_{l}(x)$ as in Theorem \ref{th:tiled-size}.
We will also introduce the distances $\widetilde{d}_{i}^{(p)}$ corresponding
to the matrices $\widetilde{\mathcal{H}}_{i}$, $i=1,2$.

The upper bound on the distance of the code with generators~(\ref{eq:Hyperbicycle})
is formulated in terms of thus introduced subset-distances ${d}_{i}^{(p)}$,
$\widetilde{d}_{i}^{(p)}$, $i=1,2$:

\begin{theorem} \label{th:upper-bound-partial} For every $p(x)$,
a binary factor of $x^{c}-1$ such that $k_{1}^{(p)}>0$ and $\widetilde{k}_{2}^{(p)}>0$,
the minimum distance $D$ of the code with generators~(\ref{eq:Hyperbicycle})
satisfies $D\le\min(d_{1}^{(p)},\widetilde{d}_{2}^{(p)})$. Similarly,
when $k_{2}^{(p)}>0$ and $\widetilde{k}_{1}^{(p)}>0$, we have $D\le\min(d_{2}^{(p)},\widetilde{d}_{1}^{(p)})$.
\end{theorem} \begin{proof} Given $k_{1}^{(p)}>0$, consider vector
$\mathbf{u}\equiv(\mathbf{e}\otimes\mathbf{c},0)$, where $\mathbf{c}\in\mathcal{C}_{\mat{H}_{1}}^{(p)}$
and $\wgt(\mathbf{e})=1$. As long as $\widetilde{k}_{2}^{(p)}>0$,
we can always select such $\mathbf{e}$ that $\mathbf{u}$ is not
a linear combination or rows of $G_{Z}$, which would indicate that
$D\le\wgt(\mathbf{c})$.

Indeed, by construction, vector $\mathbf{c}$ can be written in the
form~(\ref{eq:h1-expansion}); let us pick a bit $s$ which is not
identically zero in all $\alpha_{i}$ and construct a vector {[}cf.\ Eq.\
 (\ref{eq:h1h2-expansion}){]} 
\begin{equation}
\mathbf{u}^{(p),s}=(\underbrace{\mathbf{e}\otimes\sum_{i}\alpha_{is}I_{i}\mathbf{g}}\otimes\mathbf{e}_{s}^{1},0).\label{eq:ups}
\end{equation}
Taking all $r_{2}$ different vectors $\mathbf{e}$ and all
$k_{0}^{(p)}$ linearly-independent translations {[}Lemma
  \ref{th:lemma-H1}{]}, we obtain the vector space [as indicated in
  Eq.~(\ref{eq:ups}) with a brace] isomorphic to that on which the
subcode $\mathcal{C}_{\widetilde{\mathcal{H}}_{2}}^{(p)}$ operates. On
the other hand, there are only
$r_{2}k_{0}^{(p)}-\widetilde{k}_{2}^{(p)}$ linearly-independent
combinations of rows of the matrix $\widetilde{\mathcal{H}}_{2}$
restricted to the subspace $\mathcal{C}^{(p)}$. Since
$\widetilde{k}_{2}^{(p)}>0$, at least one of vectors
$\mathbf{u}^{(p),s}$ is linearly independent of the rows of the matrix
$\widetilde{\mathcal{H}}_{2}$ restricted to the subspace
$\mathcal{C}^{(p)}$.

Now, we can construct such a vector $\mathbf{u}$ for every $\mathbf{c}$
from the set in Eq.~(\ref{eq:subset-distance-limits}), which proves
$D\leq d_{1}^{(p)}$. The other bounds in the Theorem can be obtained
from this one by considering isomorphic codes {[}e.g., interchanging
$\widetilde{\mathcal{H}}_{2}$ and $\mathcal{H}_{1}$, and also $\mathcal{H}_{1}$
and $\mathcal{H}_{2}${]}. \end{proof}

The meaning of the condition on $p(x)$ in Theorem~\ref{th:upper-bound-partial}
can be elucidated if we rewrite the number of encoded qubits (\ref{eq:parameters})
with the help of identity (\ref{eq:subcode-transposed-dim}), 
\begin{equation}
K=\sum_{l}k_{1}^{(p_{l})}\widetilde{k}_{2}^{(p_{l})}/k_{0}^{(p_{l})}+\sum_{l}k_{2}^{(p_{l})}\widetilde{k}_{1}^{(p_{l})}/k_{0}^{(p_{l})}.\label{eq:parameters-symmetric}
\end{equation}
Obviously, every term in Eq.~(\ref{eq:parameters-symmetric}) giving
a non-zero contribution to $K$, also gives an upper bound on the
minimum distance of the quantum code.

\subsection{Codes with finite rate and distance scaling as square root
  of block length}

Here we show that the family of hyperbicycle codes contain $(v,h+v)$-limited
LDPC codes with the distance $D\propto\sqrt{N}$ that are distinct
from the hypergraph product codes. Let us start with a \emph{random
}$(h,v)$-regular parity check matrix of a classical LDPC code, where
$h<v$. By removing linearly dependent rows, we can form full-rank
$(h,v)$-limited parity check matrix $a_{1}$ that, along with $b_{j}=a_{1}^{T}$
($j\neq1$ and $\chi$ are arbitrary, in case when $j=1$ and $\chi=1$
we recover the hypergraph-product codes), we use in Eq.~(\ref{eq:Hyperbicycle})
in order to construct the hyperbicycle code where only one term in
each summation in Eq.~(\ref{eq:Hyperbicycle}) is taken. The rate
of the classical code defined by the parity check matrix $a_{1}$
is bounded from below, i.e. $R_{\mathrm{c}}\equiv k_{\mathrm{c}}/n_{\mathrm{c}}\geq1-h/v$.
With high probability at large $n_{\mathrm{c}}$, the classical code
will also have the relative distance in excess of some finite number
$\delta_{\mathrm{c}}$\cite{Litsyn:IEEE2002}. If the classical LDPC
code defined by $a_{1}$ has parameters $[n_{1},k_{1},d_{1}]$ then,
according to Theorem \ref{th:tiled-size}, \ref{th:upper-bound-partial}
and \ref{th:css-lower-bound-generic}, the quantum code will have
parameters $[[c(n_{1}-k_{1})^{2}+cn_{1}^{2},ck_{1}^{2},\geq d_{1}/c]]$.
It follows that a finite rate $(h,v)$-limited classical LDPC code
(defined by the parity check matrix $a_{1}$) with finite relative
distance (we expect the subset relative distance in
~(\ref{eq:subset-distance-limits}) 
to be finite as well) will correspond to a finite rate $(v,h+v)$-limited
quantum LDPC code with the distance $D\propto\sqrt{N}$.

\subsection{Codes with repeated codewords}

In some cases the distance of the hyperbicycle codes is larger than
the lower bound in Theorem \ref{th:css-lower-bound-generic}. In this
section we consider the special case of square matrices $a_{i}$,
$b_{i}$ ($r_{i}=n_{i}$), with the additional restriction that the
codes ${\cal C}_{{\cal H}_{i}}$, ${\cal C}_{\widetilde{\mathcal{H}}_{i}}$
are non-empty ($k_{i}=\widetilde{k}_{i}>0$) and contain only fully-symmetric
vectors in the symmetry class of $p(x)=1+x$. The results we proved
so far give the parameters of such codes summarized by (see also Theorem
3 in Ref.~\cite{Kovalev:ISIT2012}) \begin{consequence} \label{th:square-tile-symmetric-params}
Suppose $a_{i}$ and $b_{i}$ in Eq.~(\ref{eq:check-tiled}) are
such that $k_{i}^{(1+x)}=k_{i}>0$ and $r_{i}=n_{i}$. Then the CSS
code with generators~(\ref{eq:Hyperbicycle}) has the block length
$N=2cn_{1}n_{2}$, encodes $K=2k_{1}k_{2}$ qubits, and has the minimum
distance $D$ limited by $\lfloor d/c\rfloor\le D\le d$,
$d\equiv\min(d_{1},d_{2},\widetilde{d}_{1},\widetilde{d}_{2})$. 
\end{consequence} \begin{proof} By assumption, all vectors in the
codes ${\cal C}_{{\cal H}_{i}}$, $i=1,2$, are in the symmetry class
of $p(x)=1+x$, which corresponds to $k_{0}^{(p)}=1$ and block-symmetric
vectors in the form 
\begin{equation}
\mathbf{w}_{1}=\mathbf{g}\otimes\boldsymbol{\alpha},\quad\mathbf{w}_{2}=\boldsymbol{\beta}\otimes\mathbf{g},\label{eq:symmetric-vectors}
\end{equation}
respectively, with $\mathbf{g}=(1,\ldots,1)$ {[}see Eq.~(\ref{eq:h1-expansion}){]}.
The number of encoded qubits $K$ immediately follows from Theorem~\ref{th:tiled-size},
the block length $N$ from Eq.~(\ref{eq:block-size}), and the lower
bound on the distance from Theorem \ref{th:css-lower-bound-generic}.
Furthermore, with all vectors in the binary codes having the same
symmetry, the upper bound in Theorem~\ref{th:upper-bound-partial}
is just $D\le d$. \end{proof}

At this point we notice that the proof of the lower bound $\lfloor
d/c\rfloor$ on the distance in Theorem
\ref{th:css-lower-bound-generic} implies that there may be uncorrectable
errors of the form $\sum_{s}(\boldsymbol{\beta}_{s}\otimes{\bf
  g}_{s}\otimes\boldsymbol{\beta}_{s}',\boldsymbol{\alpha}_{s}'\otimes{\bf
  g}_{s}\otimes\boldsymbol{\alpha}_{s})$, where all $\mathbf{g}_{s}$
have $\wgt(\mathbf{g}_{s})=1$. On the other hand, if we were to
consider only fully-symmetric vectors, with
$\mathbf{g}_{s}=\mathbf{g}=(1,\ldots,1)$, the factor of $1/c$ would be
unnecessary. We formulate this result as \begin{statement} A symmetric
  vector $\mathbf{u}=(\mathbf{w}_{1},\mathbf{w}_{2})$,
  $\mathbf{w}_{i}=\sum_{s}\boldsymbol{\beta}_{s}^{i}\otimes\mathbf{g}\otimes\boldsymbol{\alpha}_{s}^{i}$
  with $\mathbf{g}=(1,\ldots,1)$, $i=1,2$, that satisfies
  $G_{X}\mathbf{u}=0$ and is linearly independent from the rows of
  $G_{Z}$, has sublattice weights $\wgt(\mathbf{w}_{i})$ either zero
  or $\ge d$. \label{th:symmetric-weight-lower-bound}
\end{statement}

Let us first consider the case $c=2$ (then $\chi$ must be equal to
$1$); we previously formulated the sufficient conditions to increased
lower distance bound as Theorem 3 in Ref.~\cite{Kovalev:ISIT2012}
which was given without a proof.

\begin{theorem} \label{th:tiled2-1} Suppose $c=2$, $a_{i}$ and $b_{i}$
  in Eq.~(\ref{eq:check-tiled}) are such that $k_{i}^{(1+x)}=k_{i}>0$,
  $r_{i}=n_{i}$ and binary codes with generator matrices $\sum a_{i}$,
  $\sum a_{i}^{T}$, $\sum b_{i}$ and $\sum b_{i}^{T}$ have distances
  at least $2$. Then the CSS quantum code with generators
  Eq.~(\ref{eq:Hyperbicycle}) has parameters
  $[[4n_{1}n_{2},2k_{1}k_{2},d]]$, where
  $d=\min(d_{1},d_{2},\widetilde{d}_{1},\widetilde{d}_{2})$.
\end{theorem} \begin{proof} In addition to what is stated in
  Consequence~\ref{th:square-tile-symmetric-params}, we only need to
  prove that $d$ is also the lower bound on the distance.  To this
  end, notice that any vector $\mathbf{u}$ such that
  $G_{X}\mathbf{u}=0$ can be decomposed as the sum of an ``actual''
  solution plus degeneracy,
  $\mathbf{u}^{(1+x)}+\boldsymbol{\gamma}^{T}G_{Z}$, where
  $\mathbf{u}^{(1+x)}\equiv (\mathbf{w}_1,\mathbf{w}_2)$ is a
  block-symmetric vector satisfying the conditions of
  Lemma~\ref{th:symmetric-weight-lower-bound} and linearly-independent
  from the rows of $G_{Z}$. This decomposition can be verified by
  comparing $K$ with the number of linearly-independent solutions in
  the form~(\ref{eq:ups}), as well as those on the other sublattice.
  First, let us assume $\wgt(\mathbf{w}_{1})>0$ and therefore
  $\wgt(\mathbf{w}_{1})\ge d$.  We can rewrite the corresponding
  decomposition as
  $\mathbf{w}_{1}=\sum_{s}\boldsymbol{\beta}_{s}\otimes\mathbf{g}\otimes\mathbf{e}_{s}^{1}$,
  where $\mathbf{g}\equiv (1,1)$, each $\mathbf{e}_{s}^{1}$ has length
  $n_{1}$ and $\wgt(\mathbf{e}_{s}^{1})=1$, with the non-zero element
  in the position $s$; there must be at least $d/2$ non-zero vectors
  $\boldsymbol{\beta}_{s}$. The full solution including the degeneracy
  can be formally written as
  $\sum_{s}\mathbf{w}_{1s}'\otimes\mathbf{e}_{s}^{1}$, where
\begin{equation}
  \mathbf{w}_{1s}'\equiv
  \boldsymbol{\beta}_{s}\otimes\mathbf{g}+
  \boldsymbol{\gamma}'_{s}\otimes(1,0)+
  \boldsymbol{\gamma}''_{s}\otimes(1,0),\label{eq:vector-total-decomp-c2}
\end{equation}
where the sum of the last two vectors is a linear combination of rows
of $\widetilde{\mathcal{H}}_{2}$. The key to the proof is the
observation that $\boldsymbol{\gamma}'_{s}+\boldsymbol{\gamma}''_{s}$
is a linear combination of rows of $a_{0}^{T}+a_{1}^{T}$, and
therefore is in the binary code generated by $\sum_{s}a_{s}$; by
condition the corresponding weight is either zero or $\ge2$. Without
limiting generality we can drop the case
$\boldsymbol{\gamma}'_{s}=\boldsymbol{\gamma}''_{s}\neq\mathbf{0}$
which corresponds to a symmetric vector and can be included as a part
of $\mathbf{u}^{(1+x)}$. We are left with the trivial
$\boldsymbol{\gamma}'_{s}=\boldsymbol{\gamma}''_{s}=0$, in which case
$\mathbf{w}_{1s}'=\boldsymbol{\beta}_{s}\otimes(1,1)$ remains
unchanged; otherwise
$\boldsymbol{\gamma}'_{s}\neq\boldsymbol{\gamma}''_{s}$, in which case
the weight of the modified $\mathbf{w}_{1s}'$ can be lower bounded by
that of the sum of the components corresponding to $(1,0)$ and
$(0,1)$,
\begin{equation}
  \wgt(\mathbf{w}_{1s}')\ge\wgt(\boldsymbol{\gamma}_{s}'
  +\boldsymbol{\gamma}_{s}'')\ge2;\label{eq:weight-lower-bound-c2} 
\end{equation}
with at least $d/2$ such terms the total weight is $d$ or greater.
The same arguments can be repeated in the case
$\wgt(\mathbf{w}_{2})\neq0$, as well as for the space orthogonal to
$G_{X}$. Overall, this proves the lower bound $D\ge d$; combined with
the upper bound we get $D=d$.
\end{proof}

\begin{theorem} \label{th:tiled3} Suppose $c$ is even, $a_{i}$ and
  $b_{i}$ in Eq.~(\ref{eq:check-tiled}) are such that
  $k_{i}^{(1+x)}=k_{i}$, $r_{i}=n_{i}$ and binary codes with generator
  matrices $\sum a_{i}$, $\sum a_{i}^{T}$, $\sum b_{i}$ and $\sum
  b_{i}^{T}$ have distance at least $2$. Then the quantum code in
  Eq.~(\ref{eq:Hyperbicycle}) has parameters
  $[[2n_{1}n_{2}c,2k_{1}k_{2},D]]$, where $(2/c)d\le D\le d$ and
  $d\equiv\min(d_{1},d_{2},\widetilde{d}_{1},\widetilde{d}_{2})$.
\end{theorem} \begin{proof}The proof is similar to the proof of
  Theorem \ref{th:tiled2-1}, except that now vectors
  $\mathbf{w}_{1s}'$ are defined by the analog of
  Eq.~(\ref{eq:vector-total-decomp-c2}) which has
  $\mathbf{g}=(1,\ldots,1)$ with $c$ components and more terms with
  $\gamma_s^{(j)}$ in the r.h.s., $j=1,\ldots,c$.  We need to show
  that a non-zero $\mathbf{w}_{1s}'$ has $\wgt(\mathbf{w}_{1s}')\ge
  2$, which ensures that the minimum distance of the code is at least
  $2d/c$.

  With $c>2$ and even, after the summation over all possible shifts of
  the vector $\mathbf{w}_{1s}'$ with respect to the block structure
  the symmetric term disappears, and we obtain the inequality
  $c\wgt(\mathbf{w}_{1s}')\ge c\wgt(\boldsymbol{\gamma}^{1}_{s}
  +\boldsymbol{\gamma}^{2}_{s} +\ldots+\boldsymbol{\gamma}^{c}_{s})$.
  The sum in the r.h.s.\ is a linear combination of rows of $\sum
  a_{i}^{T}$; by assumption, it's weight is either $\ge 2$ or zero.
  The only non-trivial situation corresponds to the latter case with
  some $\boldsymbol{\gamma}_s^{\ell_1}\neq 0$.  For the sum to be
  zero, either there is an even number $m$ of identical vectors
  $\boldsymbol{\gamma}_s^{\ell_1}=\boldsymbol{\gamma}_s^{\ell_2}=\ldots
  =\boldsymbol{\gamma}_s^{\ell_m}$, with $m<c$ and all indices
  different [this situation results in $\wgt(\mathbf{w}_{1s}')\ge
    (c-m)\ge 2$ since both $m$ and $c$ are even and
    $\boldsymbol{\beta}_s\neq\mathbf{0}$], or there are at least two
  pairs of unequal vectors
  $\boldsymbol{\gamma}_s^{\ell_1}\neq\boldsymbol{\gamma}_s^{\ell_2}$
  and
  $\boldsymbol{\gamma}_s^{\ell_3}\neq\boldsymbol{\gamma}_s^{\ell_4}$,
  with $\boldsymbol{\gamma}_s^{\ell_2}\neq
  \boldsymbol{\gamma}_s^{\ell_4}$, which also gives
  $\wgt(\mathbf{w}_{1s}')\ge 2$. \end{proof}

In order to obtain codes with repeated structure (see
Fig.~\ref{fig:MultiCode}), one can start with two cyclic LDPC codes
with block lengths $n_{i}$, $i=1,2$, and the check polynomials
$h_{i}(x)$ that divide $x^{n_{i}}-1$.  The polynomials $h_{i}(x)$ will
also divide $x^{cn_{i}}-1$, thus the corresponding circulant
parity-check matrix $\mathcal{H}_{i}$ of dimensions $cn_{i}\times
cn_{i}$ will lead to a code with repeated structure satisfying Theorem
\ref{th:tiled3} since the corresponding generator polynomial is
$g_{i}(x)=(x^{(c-1)n_{i}}+x^{(c-2)n_{i}}+\ldots+1)\,(x^{n_{i}}+1)/h_{i}(x)$,
$i=1,2$.

\begin{example}Suppose we use the polynomial $h(x)$ corresponding to
  the shortened Reed-Muller cyclic code with parameters
  $[2^{m}-1,m+1,2^{m-1}-1]$ in order to construct circulant matrices
  $\mathcal{H}_{1}=\mathcal{H}_{2}$ of dimensions
  $2(2^{m}-1)\times2(2^{m}-1)$. According to Theorem \ref{th:tiled3},
  a code in Eq.~(\ref{eq:Hyperbicycle}) with $c=2$ and $\chi=1$ will
  have parameters $[[4(2^{m}-1)^{2},2(m+1)^{2},2(2^{m-1}-1)]]$.  This
  family leads to weight limited LDPC codes and up to $m=11$ there is
  always a choice of polynomial $h(x)$ of weight $4$ which leads to
  quantum LDPC code with stabilizer generators of weight
  $8$.\end{example}

\begin{example}
\label{ex:Rate} 
  Given two ``small'' cyclic codes $[n_i,k_i,d_i]$ with check
  polynomials $h_i(x)$, $i=1,2$, we can construct a $c=1$
  hypergraph-product quantum code with the parameters
  $[[2n_1n_2,2k_1k_2,d]]$, $d=\min(d_1,d_2)$, a repeated even-$c$
  code with the parameters   $[[2c n_1n_2,2k_1k_2,D]]$, $2d\le D\le
  cd$, or a hypergraph-product code $[[2c^2n_1n_2, 2k_1k_2,dc]]$ using
  the ``large'' cyclic codes with the same check polynomials and the
  block lengths $cn_i$.
\end{example}
Note that in this Example the code rate goes down compared to the hypergraph-product code
constructed from the ``small'' cyclic codes and goes up compared to the hypergraph-product code constructed from the ``large'' cyclic codes.

\begin{figure}[htbp]
  \includegraphics[width=1\columnwidth]{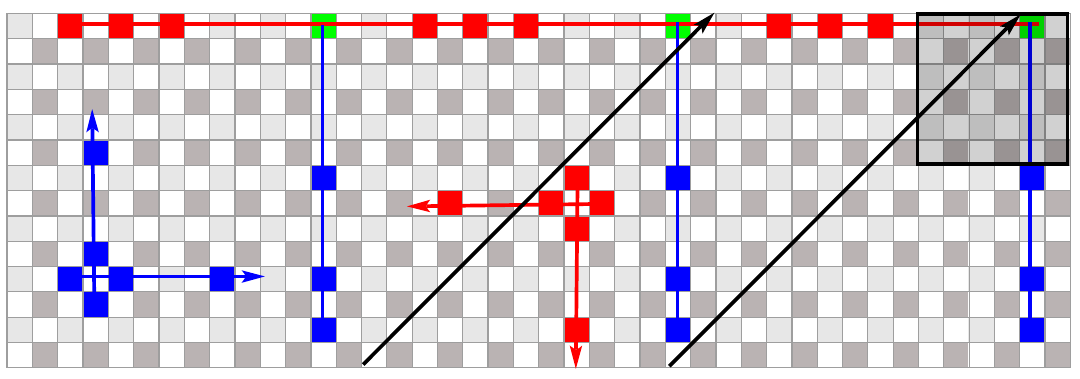} \caption{(Color
    online) Same as Fig.\ \ref{fig:Visualization} for the
    $[[294,18,4\leq D\leq12]]$ code in Eq.~(\ref{eq:Hyperbicycle})
    formed by circulant matrices $\mathcal{H}_{1}=\mathcal{H}_{2}$
    corresponding to coefficients of the polynomial $h(x)=1+x+x^{3}$
    with $c=3$ and $\chi=1$.  Two stabilizer generators are marked by
    red and blue arrows, respectively, and two anticommuting logical
    operators are marked by red and blue lines, respectively.  All
    other stabilizer generators are obtained by shifts over the same
    sublattice with periodicity in the horizontal direction and
    shifted periodicity (shown by arrows) in the vertical
    direction. In the shaded region, each gray square uniquely
    corresponds to a different logical operator, thus $18$ encoded
    logical qubits. One can observe the tripling of the logical
    operators, thus the overlap (green square) is also repeated three
    times.}

\label{fig:MultiCode} 
\end{figure}

\subsection{Planar qubit layout of hyperbicycle codes and encoding}

The stabilizer generators corresponding to Eq.~(\ref{eq:Hyperbicycle})
can be graphically represented on two rectangular regions corresponding
to two sublattices. In case, when matrices $\mathcal{H}_{1}$ and
$\mathcal{H}_{2}$ are square, the rectangular regions of sublattices
have the same dimensions and can be drawn together with parameters
$c$ and $\chi$ corresponding to the number of square blocks and
boundary shift, respectively, see, e.g., Fig. \ref{fig:Transformation}.
Furthermore, in some cases, we can represent logical operators by
line-like operators with a possibility of using this layout for encoding.
\begin{figure}[htbp]
\centering 
\includegraphics[width=1\columnwidth]{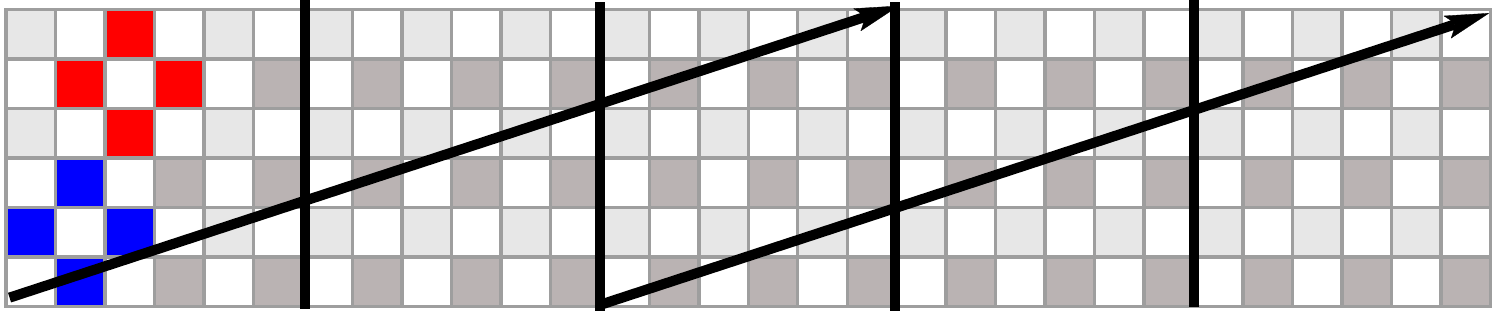}\vskip0.1in
\includegraphics[width=1\columnwidth]{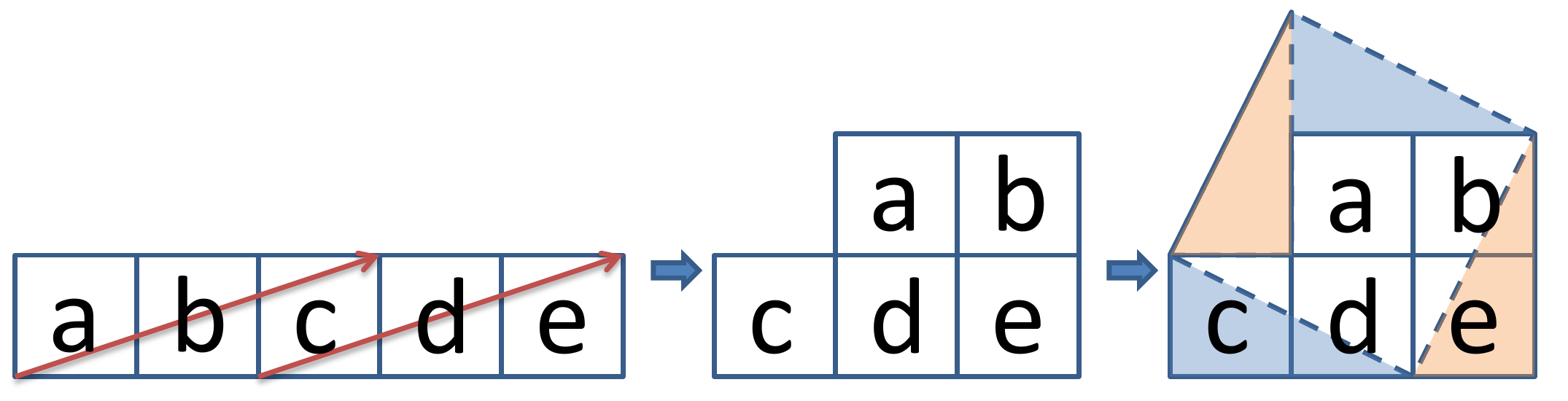} 
\caption{(Color online) Upper plot: visualization of a $[[90,2,9]]$ toric
(hyperbicycle) code in Eq.~(\ref{eq:Hyperbicycle}) formed by circulant
matrices $\mathcal{H}_{1}$ and $\mathcal{H}_{2}$ corresponding to
coefficients of the polynomial $h(x)=1+x$, $c=5$ and $\chi=3$. The
boundaries are periodic in the horizontal direction and shifted (as shown
by arrows) in the horizontal direction by $\chi=3$ blocks. Lower plot:
general block construction leading to rotated periodic boundaries
of hyperbicycle codes for $c=5$ blocks and for the shift $\chi=3$.
This corresponds to $t=1$ case of the infinite series of block constructions
with $c=t^{2}+(t+1)^{2}$ and $\chi=2t+1$, e.g., in case of toric
code we get $[[2n^{2}c,2,n\chi]]$ for any integer $n>1$.}
\label{fig:Transformation} 
\end{figure}

We start by considering the case $c=1$ and $\chi=1$ corresponding to
the hypergraph-product codes. The stabilizer generators for the
quantum code in Eq.~(\ref{eq:Till}) can be graphically represented by
two (dotted) lines living on different sublattices with the dots (red
and blue squares in Fig. \ref{fig:Visualization} marked by arrows)
placed in the positions corresponding to $1$s in rows of the binary
matrices $\mathcal{H}_{1}$, $\mathcal{H}_{2}$,
$\widetilde{\mathcal{H}}_{1}=\mathcal{H}_{1}^{T}$ and
$\widetilde{\mathcal{H}}_{2}=\mathcal{H}_{2}^{T}$. For cyclic codes,
e.g., in Fig. \ref{fig:Visualization}, the relative position of dots
stays the same and we can translate each stabilizer generator over the
corresponding sublattice.  In general, the form of stabilizer
generators is position dependent and the peculiar two-line structure
(see Fig. \ref{fig:Visualization}) ensures commutativity.  The logical
operators $\overline{X}_{j}$, $\overline{Z}_{j}$, $j=1,...,k$ can be
chosen among the rows of the matrices
$\overline{\mathcal{X}}_{1}=(\mathcal{H}_{2}^{T\perp}\otimes\widetilde{E}_{1},0)$,
$\overline{\mathcal{X}}_{2}=(0,\widetilde{E}_{2}\otimes\mathcal{H}_{1}^{T\perp})$
and
$\overline{\mathcal{Z}}_{1}=(E_{2}\otimes\mathcal{H}_{1}^{\perp},0)$,
$\overline{\mathcal{Z}}_{2}=(0,\mathcal{H}_{2}^{\perp}\otimes E_{1})$
where the index corresponds to the sublattice number on which the
logical operator lives, $\perp$ stands for the orthogonal space
$\mod2$ and matrices $\mathcal{H}_{1}^{\perp}$,
$\mathcal{H}_{2}^{\perp}$, $\mathcal{H}_{1}^{T}{}^{\perp}$ and
$\mathcal{H}_{2}^{T}{}^{\perp}$ are in a row echelon form. By row and
column permutations on matrices $\mathcal{H}_{1}$, $\mathcal{H}_{2}$
it is convenient to reduce matrices $\mathcal{H}_{1}^{\perp}$,
$\mathcal{H}_{2}^{\perp}$, $\mathcal{H}_{1}^{T}{}^{\perp}$ and
$\mathcal{H}_{2}^{T}{}^{\perp}$ to the form with an identity matrix on
the right. In such a case, the logical operators can be represented by
vertical and horizontal (dotted) lines that have only one non-zero
element in the region of the size $k_{1}\times\tilde{k}_{2}$ for the
first sublattice and of the size $\tilde{k}_{1}\times k_{2}$ for the
second sublattice (shaded region in Fig. \ref{fig:Visualization})
resulting in $k=k_{1}\tilde{k}_{2}+\tilde{k}_{1}k_{2}$ logical qubits.
Thus, for such a representation, each physical qubit in the region of
size $k_{1}\tilde{k}_{2}+\tilde{k}_{1}k_{2}$ (shaded region in Fig.
\ref{fig:Visualization}) overlaps with only one logical qubit and can
be used for encoding. Note that in general the two sublattices cannot
be drawn together as they will have different dimensions for
non-square matrices $\mathcal{H}_{1}$ and $\mathcal{H}_{2}$. In such a
case, the sublattices can be represented by two different rectangular
regions and the stabilizer generators have one line per sublattice.

The hyperbicycle construction in Eq.~(\ref{eq:Hyperbicycle}) for
arbitrary $c$ and $\chi$ has a block structure of several rectangular
regions stitched together with one of the periodic boundaries being
shifted by $\chi$ blocks (see Fig.\ \ref{fig:Transformation}). The
stabilizer generators can be graphically represented by two (dotted)
lines with the dots (red and blue squares in
Fig. \ref{fig:Visualization-1}) placed in the positions corresponding
to $1$s in rows of the binary matrices $\mathcal{H}_{1}$,
$\mathcal{H}_{2}$, $\widetilde{\mathcal{H}}_{1}$ and
$\widetilde{\mathcal{H}}_{2}$. For cyclic codes, e.g., in Fig.
\ref{fig:Visualization-1}, the relative position of dots stays the
same and we can translate the stabilizer generator with (shifted)
periodic boundaries. Just like for the hypergraph product codes, the
form of stabilizer generators is position dependent in case of
non-cyclic codes. In general, codes with $c>1$ have complicated
structure of logical operators. Nevertheless, in a specific case of
CSS codes when $c$ is odd and $k_{i}^{(1+x)}=k_{i}$ (see Theorem
\ref{th:tiled-size}), we can recover the form of logical operators
$\overline{X}_{j}$, $\overline{Z}_{j}$, $j=1,...,k$ obtained in $c=1$
case where the operators can be chosen among the rows of the matrices
$\overline{\mathcal{X}}_{1}=(\widetilde{\mathcal{H}}_{2}^{\perp}\otimes\widetilde{E}_{1},0)$,
$\overline{\mathcal{X}}_{2}=(0,\widetilde{E}_{2}\otimes\widetilde{\mathcal{H}}_{1}^{\perp})$
and
$\overline{\mathcal{Z}}_{1}=(E_{2}\otimes\mathcal{H}_{1}^{\perp},0)$,
$\overline{\mathcal{Z}}_{2}=(0,\mathcal{H}_{2}^{\perp}\otimes E_{1})$.
The only difference is that the logical operators are now repeated $c$
times which can lead to codes with increased distance (see Fig.
\ref{fig:MultiCode}).

\begin{example}A CSS $[[900,50,14]]$ hyperbicycle code is obtained
  with circulant $\mathcal{H}_{1}=\mathcal{H}_{1}^{0}$  corresponding to the
  polynomial $h(x)=(1+x+x^{3}+x^{5})$, $n_i=15$, $c=2$, $\chi=1$, and $b_i=a_i$. \end{example}

Note that one-to-one correspondence between a set of physical qubits
(shaded region in Figs. \ref{fig:Visualization} and \ref{fig:MultiCode})
and logical qubits can be used for encoding.
\begin{figure}[htbp]
\centering \includegraphics[width=0.48\columnwidth]{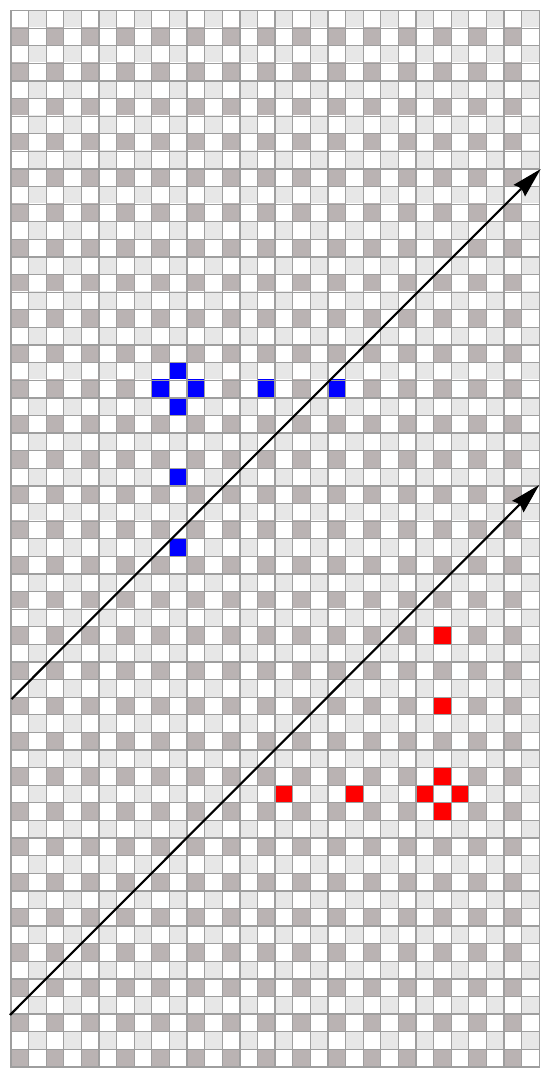}\hskip0.1in
\includegraphics[width=0.48\columnwidth]{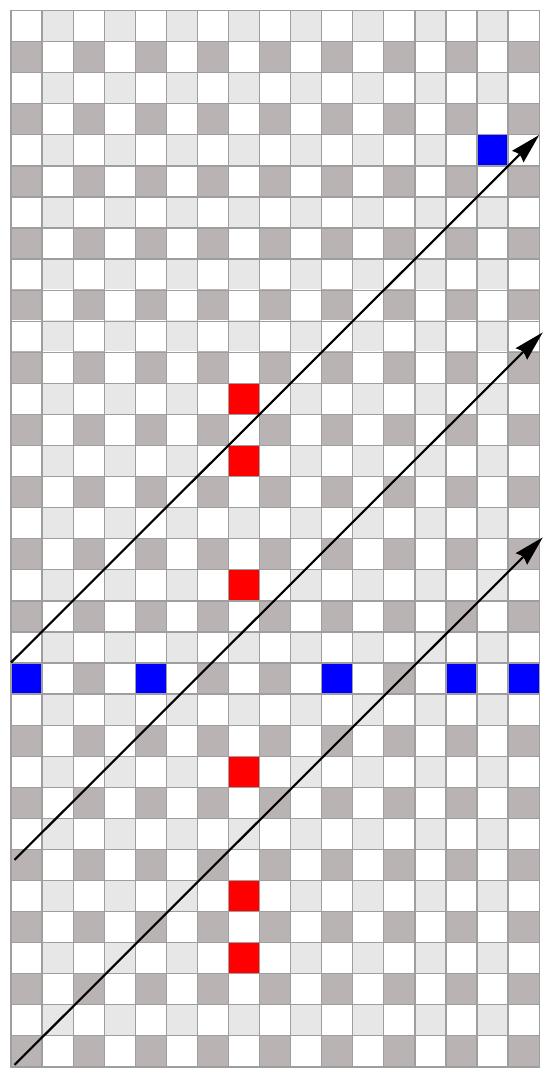}
\caption{(Color online) Same as Fig. \ref{fig:Visualization}. Left: $X$ and
$Z$ stabilizer generators for the CSS hyperbicycle code $[[900,50,14]]$
formed by circulant matrices corresponding to coefficients of a polynomial
$h(x)=1+x+x^{3}+x^{5}$ and $c=2$, $\chi=1$. Right: a single stabilizer
generator of a $[[289,81,5]]$ non-CSS hyperbicycle code in Eq.~(\ref{eq:two-sublattice})
formed by circulant matrices corresponding to coefficients of a polynomial
$h(x)=1+x+x^{3}+x^{6}+x^{8}+x^{9}$ and $c=1$. The division into
two sublattices is impossible and all other stabilizer generators
are obtained by shifts over the light and dark gray qubits with periodicity
in the vertical direction and shifted periodicity (shown by arrows)
in the horizontal direction. }

\label{fig:Visualization-1} 
\end{figure}

\subsection{Codes from two circulant matrices}

The hyperbicycle construction in Eq.~(\ref{eq:Hyperbicycle}) can
employ the known families of cyclic codes when matrices $\mathcal{H}_{1}^{0}$
and $\mathcal{H}_{2}^{0}$ (after additional permutations that change
the order in the Kronecker product) in Eq. (\ref{eq:check-tiled-1})
are circulant. Note that any circulant matrix will have the block form of
Eq. (\ref{eq:check-tiled-1}). As was mentioned in the previous
section, for circulant matrices $\mathcal{H}_{1}^{0}$ and $\mathcal{H}_{2}^{0}$
the stabilizer generators are translationally invariant with (shifted)
periodic boundaries.

The choice of $\chi\neq1$ can lead to codes with increased distance.
This can be best seen on the example with the toric code
(Fig. \ref{fig:Transformation}) where by rearranging the surface of
the code we can bring it into a new layout with proper periodic but
rotated boundaries. Then the Manhattan distance (defined on blocks,
e.g., $3$ in Fig.~\ref{fig:Transformation}) between the boundaries
will actually determine the distance of the code. The largest distance
can be expected for squares thus by defining the boundary angle as
$(t,t+h)$ we arrive at codes with $c=t^{2}+(t+h)^{2}$, $\chi=2t+h$ and
the Manhattan distance equal to $\chi$. Compared to the general
distance bound in Theorem \ref{th:css-lower-bound-generic} for toric
codes we achieve the distance: $D=\chi d/c$. As the following examples
confirm, numerically we see that $\chi>1$ can produce codes exceeding
the distance bound in Theorem \ref{th:css-lower-bound-generic}, often
saturating the upper distance bound in Theorem
\ref{th:upper-bound-partial}.

\begin{example}A CSS family of rotated toric codes is obtained when
$\mathcal{H}_{1}^{0}$ corresponds to the polynomial $h(x)=(1+x)$
(for $\mathcal{H}_{2}^{0}$ we use $b_{i}=a_{i}$), $c=t^{2}+(t+1)^{2}$,
$\chi=2t+1$, $t=1,2,\ldots$. By construction in Eq.~(\ref{eq:Hyperbicycle})
we obtain codes with parameters $[[2n^{2}c,2,n\chi]]$. Explicitly
for $n=2$ we obtain $[[40,2,6]]$, $[[104,2,10]]$\ldots{}, and
for $n=3$ $[[90,2,9]]$, $[[234,2,15]]$ \ldots{}. \end{example}

\begin{example}A $[[90,8,8]]$ CSS hyperbicycle code is obtained
when $\mathcal{H}_{1}^{0}$ corresponds to the classical cyclic code
$[15,4,8]$ with the generator polynomial $g(x)=(1+x^{3}+x^{4})$
(for $\mathcal{H}_{2}^{0}$ we use $b_{i}=a_{i}$), $c=5$ and $\chi=3$.
\end{example}

\begin{example}A $[[90,10,7]]$ CSS hyperbicycle code is obtained
when $\mathcal{H}_{1}^{0}$ corresponds to the classical cyclic code
$[15,5,7]$ with the check polynomial $h(x)=(1+x+x^{3}+x^{5})$
(for $\mathcal{H}_{2}^{0}$ we use $b_{i}=a_{i}$), $c=5$ and $\chi=3$.
\end{example}

\begin{example}
  A $[[126, 8, 10]]$ CSS hyperbicycle code is obtained when
  $\mathcal{H}_{1}^{0}$ corresponds to the classical cyclic code $[21,
    5, 10]$ with the check polynomial $h(x) = (1 + x +x^5)$ (for
  $\mathcal{H}_{1}^{0}$ we use $b_i = a_i$), $c = 7$ and $\chi = 3$.  Same
  construction  with $\chi=1$ results in the code $[[126,14,6]]$.
\end{example}

\begin{example}
  Same construction starting with  the classical cyclic code
$[30, 8, 8]$ with the check polynomial $h(x) = (1 + x^2 +x^8)$,
$c=10$ and $\chi=3$ gives a code $[[180, 16, 8]]$, while $\chi=1$
gives $[[180,16,6]]$ with a smaller distance.
\end{example}

\begin{example}
  Same construction starting with the classical cyclic code $[[30, 8,
      8]]$ corresponding to the check polynomial $h(x) = (1 + x^2
  +x^8)$ with $ c = 15$ and $\chi = 2$ gives a $[[120, 32, 4]]$ CSS
  hyperbicycle code; $\chi=1$ gives a code $[[120,32,2]]$.
\end{example}

Note that in many cases the code rate goes up compared to the hypergraph-product code
constructed from the same cyclic codes while the construction from the ``small'' cyclic codes
is not possible (cf. Example \ref{ex:Rate}).

\subsection{Non-CSS versions of hyperbicycle codes}

\label{sec:non-CSS-hb}

We observe that when $\mathcal{H}_{1}=\widetilde{\mathcal{H}}_{1}$ and
$\mathcal{H}_{2}=\widetilde{\mathcal{H}}_{2}$, the construction in
Eqs.~(\ref{eq:Hyperbicycle}) can be mapped to non-CSS codes in
Eq.~(\ref{eq:two-sublattice}) that in many cases have the same
distance but half the number of encoded and physical qubits. In
particular, this happens when $\chi=1$ and matrices $\mathcal{H}_{1}$
and $\mathcal{H}_{2}$ are symmetric. By non-CSS hyperbicycle codes we
then mean a result of the mapping in Theorem
\ref{th:two-sublattice-theorem} of the code in
Eq.~(\ref{eq:Hyperbicycle}). The dimensions of such codes can be
readily found by applying Theorem \ref{th:tiled-size} where
$s_{1}=s_{2}=0$.\begin{theorem}\label{th:Non-CSS} A quantum non-CSS
  code constructed from matrices~(\ref{eq:check-tiled-2}) such that
  $\mathcal{H}_{1}=\widetilde{\mathcal{H}}_{1}$ and
  $\mathcal{H}_{2}=\widetilde{\mathcal{H}}_{2}$ and the stabilizer
  generator matrix
\begin{equation}
G=(E_{b}\otimes\mathcal{H}_{1}|\mathcal{H}_{2}\otimes
E_{a}),\label{eq:two-sublattice-1} 
\end{equation}
 encodes $K=\sum_{l}k_{1}^{(p_{l})}k_{2}^{(p_{l})}/k_{0}^{(p_{l})}$
logical qubits into $N=cn_{1}n_{2}$ physical qubits, where $p_{l}(x)$
are all binary factors of $x^{c}-1$ such that $k_{0}^{(p_{l})}\neq0$,
including $x^{c}-1$ itself. The distance of such a code is bounded
by $D\geq\lfloor d/c\rfloor,\quad d\equiv\min(d_{1},d_{2})$ (same
notations as in Theorem \ref{th:css-lower-bound-generic}). \end{theorem}
\begin{proof}The distance bound follows from the proof of Theorem
\ref{th:css-lower-bound-generic} given the fact that any code word
of the original quantum code has to have support on at least one of
the sublattices with weight exceeding $\lfloor d/c\rfloor$.\end{proof}
\begin{theorem} \label{th:tiled3-1} Suppose $c$ is even, $a_{i}$
and $b_{i}$ in Eq.~(\ref{eq:check-tiled}) are such that $k_{i}^{(1+x)}=k_{i}$,
$r_{i}=n_{i}$ and binary codes with generator matrices $\sum a_{i}$
and $\sum b_{i}$ have distance at least $2$. Then quantum non-CSS
code with generators in Eq.~(\ref{eq:Hyperbicycle}) that have been
reduced by construction in Eq.~(\ref{eq:two-sublattice}) has parameters
$[[n_{1}n_{2}c,k_{1}k_{2},\geq(2/c)d]]$ where $d=\min(d_{1},d_{2})$.
\end{theorem} \begin{proof}This distance bound follows from the
proof of Theorem \ref{th:tiled2-1} given the fact that any code word
of the original quantum code has to have support on at least one of
the sublattices with weight exceeding $(2/c)d$.\end{proof} Finally,
we would like to mention that the upper distance bound in Theorem
\ref{th:upper-bound-partial} also applies to non-CSS hyperbicycle
codes since by construction this bound involves only one sublattice. 

For $\chi=1$ we can use \emph{palindromic} check polynomials $h(x)$,
i.e. $x^{\deg h(x)}h(1/x)=h(x)$, such that $cn-\deg h(x)$ is even,
in order to construct symmetric circulant matrices $\mathcal{H}_{i}$
from the polynomial $x^{[cn-\deg h(x)]/2}h(x)$. 

\begin{example}A $[[289,81,5]]$ non-CSS hyperbicycle code (see Fig.
\ref{fig:Visualization-1}) is obtained from Eqs.~(\ref{eq:two-sublattice})
and ~(\ref{eq:Hyperbicycle}) using circulant matrices $\mathcal{H}_{1}=\mathcal{H}_{2}$
corresponding to coefficients of a palindromic polynomial $h(x)=1+x+x^{3}+x^{6}+x^{8}+x^{9}$
where $c=1$ and $\chi=1$. \end{example}

\section{Conclusions}

We described a large family of hyperbicycle codes that includes as
subclasses the best of the known LDPC codes.  The construction allows
for explicit upper and lower bounds on the code distance.  We also
described new LDPC code families with finite rates and distances
scaling as a square root of block length. Our discussion is
accompanied with geometrical interpretations of the hyperbicycle codes
which can facilitate design and applications of such codes.  The
construction is particularly important for designing LDPC codes with
relatively small block lengths which is important since the original
hypergraph product codes have relatively poor parameters at small
block lengths.

Another advantage of hyperbicycle construction is that it can be based
on a pair of very well studied classical cyclic codes. This leads to
codes with good parameters up to limited but relatively large block
lengths (in general, cyclic codes with asymptotic rates below one have
poor asymptotic parameters).  The planar layout of thus constructed
quantum codes possess translational invariance of stabilizer
generators which may simplify the implementation (see e.g. Ref. \cite{De:arXiv2012}).

Although the quantum LDPC codes discussed in this work have been
shown to possess a finite noise threshold \cite{Kovalev:2012arXiv}, it
is yet to be seen whether there are good decoders for such codes. It
may well happen that the relation between hyperbicycle and bicycle
codes can lead to better decoding for the former as the latter are
known for their good decoding properties.

Even though the lower distance bounds presented in this paper are in
some cases inferior compared to the hypergraph-product codes, we do
not expect that this will have a significant effect on the value of the
noise threshold as the distance still scales as a square root of the
block length while the LDPC structure of the stabilizer generators is
preserved \cite{Kovalev:2012arXiv}.  Given that, we expect that one
can encode more qubits into hyperbicycle codes compared to hypergraph
product codes without affecting the threshold.

Our results notwithstanding, there are several open questions in regard to the hyperbicycle codes.
In particular, it would be interesting to establish conditions
under which the hyperbicycle codes reach the upper distance bound. Furthermore, the case when the block shift $\chi$ and the number of blocks $c$ are commensurate 
has not been analyzed.  It
would also be interesting to explore the exact relation between the
hyperbicycle codes and the CSS codes constructed over higher alphabets
\cite{Tillich:ISIT2012,Kasai:IEEE2012}.
\newline

\section*{ACKNOWLEDGEMENTS}

We are grateful to I. Dumer and M. Grassl for multiple helpful discussions.
This work was supported in part by the U.S. Army Research Office under
Grant No.\ W911NF-11-1-0027, and by the NSF under Grant No. 1018935.

\bibliography{more_qc,MyBIB,qc_all}

\end{document}